\documentclass[a4paper,11pt]{article}

\newcommand\version{0}

\usepackage{appendix}
\usepackage{amsmath}
\usepackage{amsthm}
\usepackage{amssymb}
\usepackage{amsfonts}
\usepackage{standalone}
\usepackage{bbm}
\usepackage{tikz}
\usepackage{verbatim}
\usetikzlibrary{shapes.geometric, arrows}
\usepackage[margin=2.5cm]{geometry}
\usepackage{hyperref}
\usepackage{mathtools}
\usepackage[compact]{titlesec}
\usepackage{enumitem}
\usepackage{etoolbox}
\usepackage{pgfplots}
\usepgfplotslibrary{polar}
\usepackage[caption=false]{subfig}
\usepackage{blochsphere}

\theoremstyle{plain}
\newtheorem{theorem}{Theorem}[section]
\newtheorem{proposition}[theorem]{Proposition}
\newtheorem{corollary}[theorem]{Corollary}
\newtheorem{lemma}[theorem]{Lemma}

\theoremstyle{definition}
\newtheorem{procedure}[theorem]{Procedure}
\newtheorem{construction}[theorem]{Construction}
\newtheorem{definition}[theorem]{Definition}
\newtheorem{remark}[theorem]{Remark}
\newtheorem{example}[theorem]{Example}

\def\Im{\ensuremath{\mathrm{Im}}}
\def\diag{\ensuremath{\mathrm{diag}}}
\def\End{\ensuremath{\mathrm{End}}}

\def\Tr{\ensuremath{\mathrm{Tr}}}
\def\Re{\ensuremath{\mathrm{Re}}}
\def\T{\ensuremath{\mathrm{T}}}

\newcounter{jamiecomment}
\newcommand\JVcomm[1]{\ensuremath{{}^{\color{red}\thejamiecomment}}\marginpar{\color{red}\tiny\raggedright \thejamiecomment: #1}\stepcounter{jamiecomment}}
\newcounter{dominiccomment}
\newcommand\DVcomm[1]{\ensuremath{{}^{\color{blue}\thedominiccomment}}\marginpar{\color{blue}\tiny\raggedright \thedominiccomment: #1}\stepcounter{dominiccomment}}
\newcommand\ignore[1]{}
\def\Z{\mathbb{Z}}
\def\R{\mathbb{R}}
\def\C{\mathbb{C}}

\def\SU{\ensuremath{\mathrm{SU}}}
\def\U{\ensuremath{\mathrm{U}}}
\def\SO{\ensuremath{\mathrm{SO}}}

\def\B{\ensuremath{\mathrm{B}}}

\newcommand{\ket}[1]{\,\left| #1 \right\rangle}

\def\cos{\mathrm{c}}
\def\sin{\mathrm{s}}
\def\tan{\mathrm{t}}

\newcommand{\finiteabstract}{\begin{abstract}
We present a new scheme for teleporting a quantum state between two parties whose local reference frames are misaligned by the action of a finite symmetry group. Unlike other proposals, our scheme requires the same amount of classical communication and entangled resources as conventional teleportation, does not reveal any reference frame information, and is robust against changes in reference frame alignment while the protocol is underway. The mathematical foundation of our scheme is a unitary error basis which is permuted up to a phase by the conjugation action of the group. We completely classify such unitary error bases for qubits, exhibit constructions in higher dimension, and provide a method for proving nonexistence in some cases.\end{abstract}}

\newcommand{\finiteacknowledgements}{We are grateful to Niel de Beaudrap, Simon Benjamin, Subhayan Moulik, Benjamin Musto, David Reutter, Isar Stubbe, Sean Tull and Linde Wester for useful discussions. We thank two anonymous referees for their detailed and helpful comments regarding the presentation of these results. We used the \emph{blochsphere} and \emph{solides-3d} \LaTeX\ packages.  The first author acknowledges support from the Engineering and Physical Sciences Research Council.}

\newcommand{\finitetitleetc}[1]{\ifnumcomp{#1}{=}{1}{\title{Perfect tight quantum teleportation without a shared reference frame}
\date{\today}
\author{Dominic Verdon}\email{dominic.verdon@cs.ox.ac.uk}
\author{Jamie Vicary}\email{jamie.vicary@cs.ox.ac.uk}
\affiliation{Department of Computer Science, Wolfson Building, University of Oxford, Parks Rd, Oxford OX1 3QD} \finiteabstract \maketitle}
{
\title{\vspace{-2.5cm}\bf \LARGE \mbox{Perfect tight quantum teleportation} \mbox{without a shared reference frame}}
\author{\vspace{-.70cm}
\\
Dominic Verdon\thanks{\texttt{dominic.verdon@cs.ox.ac.uk}} \ and Jamie Vicary\thanks{\texttt{jamie.vicary@cs.ox.ac.uk}}
\\
Department of Computer Science, University of Oxford\vspace{-5pt}}
\maketitle \vspace{-1cm} \finiteabstract}}

\newcommand{\finiteintroacknowledgements}[1]{\ifnumcomp{#1}{=}{1}{}{
\paragraph{Acknowledgements.}
\finiteacknowledgements}}

\newcommand{\finitepraacknowledgements}[1]
{\ifnumcomp{#1}{=}{1}{\acknowledgements{\finiteacknowledgements}}{}}
\newcommand{\biblstyle}[1]{\ifnumcomp{#1}{=}{1}{}{\bibliographystyle{plainurl}}}
\newcommand{\centeringforarxiv}[1]{\ifnumcomp{#1}{=}{1}{}{\centering}}

\begin{document}

\finitetitleetc{\version}

\section{Introduction}
\label{sec:introduction}

\paragraph{Motivation.}

It is now well recognized that a shared reference frame is an implicit assumption underlying the correct execution of many quantum protocols~\cite{Bartlett2007, Kitaev2004, Gisin2002, Verstraete2003, Gour2008, Enk2001}. As quantum communication finds its way into handheld devices~\cite{Wabnig2013, Duligall2006, Duligall2007} and into space~\cite{Ren2017, Yin2017, Bacsardi2013}, it is increasingly important to develop protocols robust against reference frame error for situations where alignment is difficult~\cite{Islam2014, Islam2016, Skotiniotis2012} or undesired~\cite{Bartlett2004, Ioannou2014}. Considerable progress has already been made in this regard for quantum key distribution~\cite{DAmbrosio2012, Zhang2014, Wang2015, Liang2014, Souza2008, Laing2010, Vallone2014}, and there is also a smaller body of work on quantum teleportation~\cite{ Chiribella2012,Marzolino2015,Marzolino2016} without a shared reference frame, which our results extend.

\paragraph{Main results.} We consider the problem of quantum teleportation between two parties whose local reference frames are misaligned, where the set of possible local reference frame transformations forms a finite group $G$ with a unitary representation $\rho:G \to \U(d)$ on the $d$-dimensional system to be teleported. (This is the first paper in a series; the second paper~\cite{VerdonInfinite} extends these results to the more common setting of infinite groups.) Success of the protocol is judged by a third-party observer who holds full reference frame information, and who must agree that the original state has been teleported perfectly up to a global phase.\footnote{This was called \textit{unspeakable quantum teleportation} by Chiribella et al~\cite{Chiribella2012}.} We present a teleportation scheme for certain $(G,\rho)$, where $G$ is finite, which is guaranteed to succeed regardless of the parties' reference frame configurations and which additionally satisfies the following properties.
\begin{itemize}
\item \label{property:tightness} \emph{Tightness}. The parties only require a  $d$-dimensional maximally entangled resource state, and only $2$ dits of classical information are communicated from Alice to Bob.
\item \label{property:dr}\emph{Dynamical robustness} (DR). The scheme is not affected by changes in reference frame alignment during transmission of the classical message from Alice to Bob.
\item \label{property:nl} \emph{No reference frame leakage} (NL). No information about either party's reference frame alignment is transmitted.\footnote{This has cryptographic significance in some scenarios~\cite{Ioannou2014, Bartlett2004, Kitaev2004}.}
\end{itemize}

Our scheme depends on the existence  of a \textit{$G$\-equivariant unitary error basis} for the representation $(G,\rho)$. We exhaustively classify these mathematical structures for two-dimensional representations, showing that they exist precisely when the image of the composite homomorphism $G \stackrel \rho \to \U(2) \stackrel q \to \SO(3)$ is isomorphic to 1, $\Z_2$, $\Z_3$, $\Z_4$, $D_2$, $D_3$, $D_4$, $A_4$ or $S_4$, where $q$ is the quotient taking a unitary to its corresponding Bloch sphere rotation. We also provide a construction for any permutation representation with dimension less than $5$, and show how to prove nonexistence in some cases.

Our results rely on a new idea regarding the classical communication part of the protocol: we suppose that the readings of the classical channel are \emph{themselves} interpreted with respect to the local reference frame. Mathematically, this corresponds to a nontrivial action of the group of reference frame transformations on the classical channel. Such classical channels have been called `unspeakable'~\cite{Peres2002}; we provide examples, and show how they  can be used to communicate the measurement result. An unspeakable classical channel is a powerful resource which could be used to execute a prior alignment step before the protocol begins, but we emphasize that it is \textit{not} being used in this way here; indeed, by the (NL) property, our protocol in fact transfers no information at all about either party's reference frame alignment, and makes use of the unspeakable channel in a nontrivial way.

We can give the following simple intuition for how our scheme works. Local reference frame misalignment can cause errors in the performance of the protocol, since Bob will perform correction operations with respect to his own frame, which need not be aligned with the frame in which Alice performed her measurement. But, since in our setting the misalignment also affects the classical channel, it can also cause errors in transmission of the classical measurement result; Bob may, in interpreting the channel reading with respect to his own frame, receive a different measurement value to that transmitted by Alice. In essence, our scheme is constructed so that these errors exactly cancel out. This intuition makes clear how the (DR) property is possible, since a change in local reference frame alignment also affects reception of the classical communication data, even if it takes place while that information is in transit. 

\paragraph{Related work.}
Chiribella et al.~\cite{Chiribella2012} considered teleportation with a speakable classical channel only, and showed that when the group $G$ of reference frame transformations is a continuous compact Lie group, perfect tight teleportation is impossible; this does not contradict our work, which uses an unspeakable classical channel and a finite group $G$. (Furthermore, as a consequence of our main results, we show that for finite $G$, perfect tight teleportation is indeed possible with a speakable classical channel in some restricted situations; see Corollary~\ref{cor:speakableonly} and Remark~\ref{rem:qubitspeakabletel}.)

Several other solutions for reference frame--independent teleportation for a finite group of reference frame transformations exist in the literature. These all involve establishment of a shared reference frame in some way: by using pre-shared entanglement~\cite{Chiribella2012}, sharing entanglement during the protocol~\cite{Kitaev2004}, or transmitting more complex resources~\cite[Section V.A]{Bartlett2007}. Unlike our scheme, these approaches work for arbitrary $(G,\rho$) where $G$ is finite. However, none of them have all the properties of tightness, dynamical robustness and no reference frame leakage, as our scheme does.

Quantum communication under collective noise corresponding to a finite group was considered by Skotiniotis et al.~\cite{Skotiniotis2013}. From the perspective of our discussion above, their protocol satisfies the (DR) and (NL) properties. However, it requires a quantum channel; it is not a teleportation protocol. Their token could be equally be transmitted using an unspeakable classical channel of the type we construct in Section~\ref{sec:unspchan}. However, we are not transmitting a token in their sense; in particular, the classical system we transmit need not carry a free and transitive action of $G$.

\paragraph{Criticism.} We can criticise our scheme as follows.
Firstly, as with the alternative solutions discussed above, it works only for finite $G$ (although we discuss a related scheme for the case of infinite $G$ in a successor article~\cite{VerdonInfinite}.) Secondly, it cannot be implemented for all scenarios $(G,\rho)$ with finite $G$, and, although we provide a range of constructions of equivariant unitary error bases, and completely characterise valid $(G,\rho)$ for qubit teleportation, we cannot give necessary and sufficient conditions for the applicability of our scheme in higher dimensions. Thirdly, to communicate the measurement result, we do not use an ordinary `speakable' classical channel, but rather an `unspeakable' classical channel; while we provide a number of examples of such channels, it is nevertheless clear that this novel aspect of our approach may raise technological barriers in an implementation.
Finally, up to a global phase, the system to be teleported and Bob's half of the entangled pair must carry the same representation $\rho$ of $G$, and Alice's half of the entangled pair must transform according to the dual representation $\rho^*$; although this is physically reasonable in view of charge conservation, a situation may arise in which it is hard to construct a system carrying the representation $\rho^*$. Very often (for instance, for all representations with real characters), $\rho \simeq \rho^*$ up to a phase, which solves this problem.

\paragraph{Outlook.} These results may be applicable to cryptography and security of quantum protocols, as it has been noted that reference frame uncertainty is of cryptographic importance \cite{Ioannou2014, Bartlett2004, Kitaev2004}, and that a private shared reference frame may be considered as a secret key~\cite{Ioannou2014, Bartlett2004}. In this context, it is useful to know what protocols, such as quantum teleportation, may be performed even in the absence of a shared reference frame, without any transmission of cryptographically sensitive reference frame information. 

We can also build on these results to produce schemes for teleportation with a continuous compact Lie group of reference frame transformations. This is treated in a forthcoming paper~\cite{VerdonInfinite}.

\paragraph{Outline.} In Section~\ref{sec:rfindeptel} we present our scheme for reference frame--independent teleportation, beginning with an informal example for a group of spatial reference frame transformations. 
Our scheme uses  an unspeakable classical channel carrying a certain action; in Section~\ref{sec:unspchan} we show how these may be constructed, and give several examples.
Finally, in Section~\ref{sec:equivuebs} we turn our attention to the problem of classifying and constructing equivariant unitary error bases, on which our scheme depends. 

\finiteintroacknowledgements{\version}

\section{Reference frame--independent teleportation}
\label{sec:rfindeptel}
\subsection{Example}
\label{sec:example}

\paragraph{Scenario.} 

Alice and Bob are quantum information theorists operating on spin-$\frac{1}{2}$ particles. They work in separate laboratories, which do not necessarily have the same orientation in space, and their task is to teleport a quantum state  without revealing their spatial orientations, either to each other or to any eavesdropper.
Their relative orientations are not completely unknown: \ignore{they are promised to lie within the subgroup $\Z_3 \subset \SO(3)$\JVcomm{Need to say something about embedding in $\U(1)$}, the group of rigid spatial rotations.} the rotation $g$ taking Alice's Cartesian frame onto Bob's is promised to lie within the subgroup $\Z_3 \subset \SO(3)$, the group of rigid spatial rotations through multiples of $2\pi/3$ radians around some axis. However, $g \in \mathbb{Z}_3$ is unknown. Let $a \in \mathbb{Z}_3$ be the transformation rotating the reference frame anticlockwise through $2\pi/3$ radians. We suppose that the action of $a$ affects the description of qubit states by the standard spin-1/2 representation:
\begin{equation}
\label{eq:exampleaction}
\rho(a) = \begin{pmatrix}
1 & 0 \\ 0 & e ^{2 \pi i/3}
\end{pmatrix}
\end{equation}
That is, a state which appears as $\ket{v}$ in frame configuration $f$ will appear as $\rho(a)\ket{v}$ in frame configuration $a \cdot f$.

Alice and Bob share the two-qubit entangled state
\begin{equation*}\ket \eta = \frac 1 {\sqrt 2} ( \ket {01} + \ket {10} ).\end{equation*} Note that this state is invariant up to a phase under the action~\eqref{eq:exampleaction} of a change in reference frame orientation, so the entanglement will not be degraded by changes in reference frame alignment following its initialisation.  All these aspects of the overall setup are common knowledge to both parties.

\paragraph{The conventional protocol.} A conventional quantum teleportation scheme~\cite{Werner2001} is presented in terms of a \emph{unitary error basis} (a family of unitary operators which form an orthogonal basis for the operator space under the trace inner product):
\begin{align}
\nonumber
U_0 &= \scriptsize \begin{pmatrix} 1 & 0 \\ 0 & e^{2 \pi i /3} \end{pmatrix}
&
U_2 &= \frac{1}{\sqrt{3}} \scriptsize\begin{pmatrix} 1 & \sqrt{2}e^{2 \pi i /3} \\ \sqrt 2 & e^{5 \pi i /3} \end{pmatrix}
\\[-9pt]
\label{eq:exampleueb}
\\[-5pt]
\nonumber
U_1 &= \frac 1 {\sqrt{3}} \scriptsize \begin{pmatrix} 1 & \sqrt{2}e^{4 \pi i/3}\\ \sqrt{2}e^{4 \pi i/3} & e^{5 \pi i/3} \end{pmatrix}
&
U_3 &= \frac {1} {\sqrt{3}} \scriptsize \begin{pmatrix} 1 & \sqrt{2} \\ \sqrt{2}e^{2\pi i/3} & e^{5 \pi i/3} \end{pmatrix}
\end{align}
The scheme proceeds as follows. Alice measures her initial system together with her half of the entangled state in a maximally-entangled orthonormal basis $\ket {\phi_i} = (\mathbbm 1 \otimes (  U_i X)^T) \ket \eta$, where $X$ is the Pauli $X$-matrix~\footnote{The Pauli $X$-matrix appears because of the choice of entangled state $\eta$.}, and communicates the result $i$  to Bob through an ordinary classical channel, which transmits the measurement result faithfully. Bob then applies the  correction $U_i$ to his half of the entangled state.

If the reference frames have the same alignment, the procedure will be successful. However, if the reference frames are misaligned by some nonidentity element $g \in \Z_3$, then, from the perspective of Alice's frame, Bob will not perform the intended correction $U_i$, but rather $\rho(g)^\dag U_i \rho(g)$. Assuming the uniform distribution over $\Z_3$, a simple calculation shows that an input pure state will emerge in a mixed state.

\paragraph{The new protocol.}  We now describe our reference frame--independent scheme. Before performing the protocol, Alice and Bob share the coordinates of four unit vectors $\{v_0,v_1,v_2,v_3\} \in \R^3$, which form a regular tetrahedron centred on the origin such that, under the reference frame transformation $a \in \Z_3 \subset \SO(3)$, the vectors are permuted as follows:
\begin{align}
\label{eq:arrowtransformation}
\hspace{-4pt}a \cdot v_0 &= v_0 & a \cdot v_1 &= v_2 & a \cdot v_2 &= v_3 & a \cdot v_3 &= v_1
\end{align}
For example, let $v_0 = \frac 1 {\sqrt 3} (\hat x + \hat y + \hat z)$, $v_1 = \frac 1 {\sqrt 3} (\hat x - \hat y - \hat z)$, $v_2 = \frac 1 {\sqrt 3} (-\hat x + \hat y - \hat z)$ and  $v_3 = \frac 1 {\sqrt 3} (-\hat x - \hat y + \hat z)$, and suppose that the generating element $a \in \Z_3$ acts as a right-handed rotation about the axis defined by $v_0$.

If Alice obtains measurement result $i$, she communicates this to Bob in the following way: she prepares a physical arrow, of the sort a medieval archer might use, arranges it to have the same orientation as the vector $v_i$, and then sends it directly to Bob by parallel transport along a known path. When the arrow is received, Bob observes its orientation in his own frame, correcting if necessary for the parallel transport map associated to the path, and matches this with one of the reference orientations $v_j \in \{v_0, v_1, v_2, v_3\}$; he thus obtains the message $j \in \{0,1,2,3\}$. He then performs the corresponding unitary correction. This procedure is illustrated in Figure~\ref{fig:exampletetrahedron}. \begin{figure}\centering \includegraphics{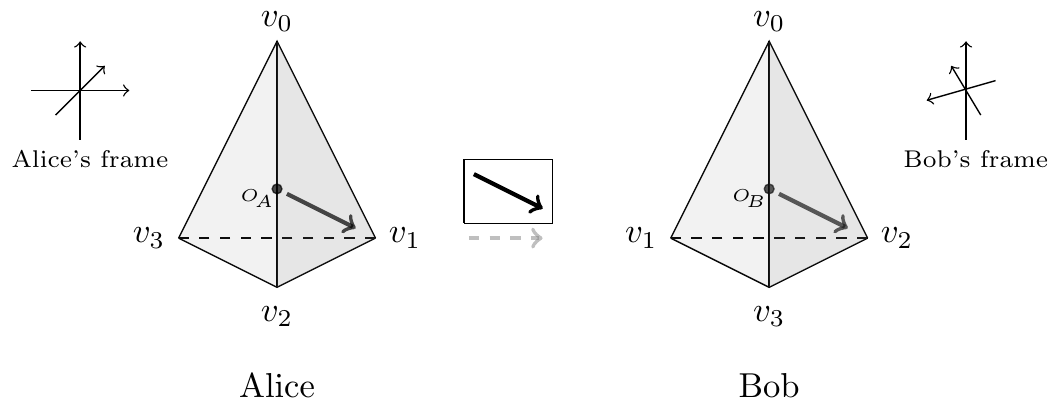}
\caption{\small In our classical communication procedure, Alice and Bob label the vertices of regular tetrahedra centred on their origins $O_A$ and $O_B$, using their own Cartesian frames. Bob's frame is related to Alice's by a $2\pi/3$ anticlockwise rotation around the axis defined by $v_0$. Upon measuring $\ket{\phi_1}$, Alice prepares an arrow pointing to vertex $v_1$ and sends this to Bob by parallel transport. In Bob's frame this arrow points to vertex $v_2$, and so he performs correction $U_2$.}\label{fig:exampletetrahedron}\end{figure} 

Note that Alice transmits no information about her local reference frame by the above procedure, since her measurement result is uniformly random, and thus so is the direction indicated by the arrow. Also, we emphasize that exactly  two bits of classical information have been transferred, since there were four possible values upon transmission and four possible values upon receipt.

Suppose that Alice and Bob's laboratories share the same reference frame; that is, their local frames are related by the element $e \in \Z_3$ of the group of reference frame transformations. Then the arrow's  orientation will be the same in Bob's frame as in Alice's frame, and the measurement outcome will be faithfully communicated. In this case the protocol will be successful, and it is identical to the conventional teleportation protocol, albeit with the two classical bits of information transmitted from Alice to Bob in an unusual way.

Now suppose that Alice and Bob's frames are misaligned by the action of the element $a \in \Z_3$ of the reference frame transformation group. In this case, if Alice sends the result 0, 1, 2, or 3, Bob will receive the result 0, 2, 3 or 1 respectively, because of the transformation properties~\eqref{eq:arrowtransformation} of the arrows. Furthermore, when Bob applies the unitary $U_i$ in his local frame, its action is seen in Alice's frame as $\rho(a) ^\dag U_i \rho(a)$. The following equations describe the consequences of such a conjugation, as can be directly checked using expressions~\eqref{eq:exampleaction} and~\eqref{eq:exampleueb}:
\begin{align*}
\rho(a)^\dag U_0 \rho(a) &= U_0
&
\rho(a)^\dag U_1 \rho(a) &= U_3
\\
\rho(a)^\dag U_2 \rho(a) &= U_1
&
\rho(a)^\dag U_3 \rho(a) &= U_2
\end{align*}
We now see the point of the entire construction: the unitary error basis~\eqref{eq:exampleueb} was carefully chosen so that these two apparent sources of error---in the transmission of the classical measurement result, and in Bob's unitary correction---exactly cancel each other out. For example, if Alice obtains measurement outcome~1, Bob will receive this as measurement outcome~$2$, and will perform the correction $U_2$ in his frame, which in Alice's frame is equal to $\rho(a)^\dag U_2 \rho(a) = U_1$, and so the intended correction will be carried out after all. As a result, the quantum teleportation will conclude successfully, even though Alice and Bob's reference frames were misaligned. Similarly, it can be shown that the teleportation is also successful if the frame misalignment is given by the element $a^2 \in \Z_3$.

\paragraph{Discussion.}

We have exhibited a procedure for reference frame--independent quantum teleportation in the particular case of spatial reference frame misalignment with transformation group $\Z_3 \subset \SO(3)$. This involved  a careful choice of unitary error basis~\eqref{eq:exampleueb}, with communication of the measurement result through a classical channel carrying a compatible nontrivial action~\eqref{eq:arrowtransformation} of the reference frame transformation group. Only 2 bits of classical information were transferred from Alice to Bob, as in a conventional teleportation procedure, and the Hilbert space of the entangled resource was of minimal dimension, so this procedure was \emph{tight} in the sense of Werner~\cite{Werner2001}. The unspeakable information transmitted by Alice was uniformly random, since Alice's measurement results were; in particular, Bob, or an eavesdropper on the classical channel, received no information about Alice's reference frame alignment. Finally, the procedure would have succeeded even if Bob's reference frame alignment changed during the protocol, while Alice's measurement result was still in transit.

In this example we chose $\Z_3 \subset \SO(3)$ as the reference frame transformation group, but the same unitary error basis and classical channel allow reference frame--independent teleportation for the group $A_4 \subset \SO(3)$ of order 12, as we will see in Section~\ref{sec:equivuebs}.

\subsection{General scheme}
\label{quantumteleportationsection}
We now present our scheme in full generality. We begin by recalling the conventional teleportation protocol.
\begin{procedure}[Conventional tight teleportation~\cite{Bennett1993}]\label{proc:oldschooltel}
Alice holds an $n$-dimensional quantum system, prepared in a state $\ket{\psi}$. Separately, Alice and Bob hold an entangled pair of $n$-dimensional quantum systems, in a maximally entangled state $(1 \otimes X)\ket{\eta}$ for some unitary $X$, where $$\ket \eta = \frac 1 {\sqrt n} \sum_{i=1}^n \ket {ii}$$ is the generalised Bell state.\footnote{All maximally entangled states of a bipartite system are of this form.} Alice performs a joint measurement on the system to be teleported and her entangled system, described by an orthonormal basis $\ket {\phi_i} \in \C^n \otimes \C^n$. She communicates the classical measurement result $i$ to Bob using a perfect classical channel; Bob then performs the unitary correction $U_i$ on his half of the entangled state. The procedure is successful if Bob's system is now in the state $\ket{\psi}$.
\end{procedure}

A complete description of correct procedures was given by Werner, who showed that they can be characterized mathematically in terms of unitary error bases.
\begin{definition}
\label{def:ueb}
For a Hilbert space $H$, a \textit{unitary error basis} (UEB) is a basis of unitary operators $\{U_i\}_{i \in I}$, with $I= \{0,1, \dots,\dim(H)^2-1\}$, such that for all $i,j \in I$ we have:
\begin{equation}
\Tr \big( U_i ^\dag U_j ^{} \big) = \delta_{ij}\dim(H)
\end{equation}
\end{definition}

\noindent
Under this correspondence, we construct Alice's joint measurement basis as 
\begin{equation}\label{eq:alicemeasbasis}\ket {\phi_i} := (\mathbbm 1 \otimes X^\T U_i^\T) \ket \eta,
\end{equation}
and Bob performs the correction $U_i$ from the unitary error basis when he receives the measurement result $i$ from Alice. 
Werner showed~\cite[Theorem 1]{Werner2001} that all correct measurement and correction data for Procedure~\ref{proc:oldschooltel} can be obtained from a unitary error basis in this way.

A second key concept in our new scheme is that of an \emph{unspeakable classical channel}. For simplicity, we only consider perfect classical channels in this paper; whatever reading Alice sends through the channel will be received unaltered by Bob. However, his interpretation of this reading will be affected by his reference frame orientation.
\begin{definition}\label{def:unspeakablecomm}
For a finite group $G$, an \textit{unspeakable classical channel} is a classical channel whose set of messages carries a nontrivial action of the group $G$ of reference frame transformations.
\end{definition}

\noindent
Writing $I$ for the set of messages carried by the channel, we can encode the data of an unspeakable channel as a group action ${\sigma}: G \times I \to I$. For each reference frame transformation $g \in G$ taking Alice's frame onto Bob's frame, we obtain an invertible function ${\sigma}(g,-): I \to I$, which describes how a message input by Alice using her local frame is interpreted by Bob with respect to his local frame. Since this function is invertible, there is no loss of information; however, if the receiver of the message does not know $g \in G$, they will be unable to infer which message was actually input. The arrows channel of Section~\ref{sec:example} was an unspeakable classical channel; we will see more examples in Section~\ref{sec:unspchan}.

We now define our new teleportation scheme. Here we write $\rho^*$ for the dual representation of $\rho$. 
\begin{procedure}[Reference frame--independent teleportation]
\label{chargeteleportationopint}
Alice has an $n$-dimensional quantum system in a state $\ket{\psi}$. Separately, Alice and Bob hold a maximally entangled state $(\mathbbm{1} \otimes X)\ket{\eta}$ of a pair of $n$-dimensional quantum systems. They each possess local reference frames with transformation group $G$, acting unitarily by a representation $\rho$ on the system to be teleported, by a representation $\rho^* \otimes \theta_1$ on Alice's half of the entangled state, and by a representation $\rho \otimes \theta_2$ on Bob's half of the entangled state, where $\theta_1, \theta_2$ are any one-dimensional representations of $G$. 

Alice performs a joint measurement on the system to be teleported and her half of the entangled state, described by an orthonormal basis $\{\ket {\phi_i}\}, \ket {\phi_i} \in \C^n \otimes \C^n$. She uses a perfect unspeakable classical channel to communicate the classical measurement result $i$ to Bob, who receives the message ${\sigma}(g,i)$, where $g$ is the transformation taking Alice's local frame configuration upon transmission onto Bob's local frame configuration upon receipt. Bob then immediately performs a unitary correction $U_{{\sigma}(g,i)}$ on his half of the entangled state. 
\end{procedure}

\begin{remark}\label{rem:entangledstatealignment}
We prove in Appendix~\ref{sec:invariantstates} that the conditions on the possible representations carried by each system precisely imply that the maximally entangled state may always be taken to be $G$-invariant up to a phase, preventing degradation of entanglement by reference frame transformations.
\ignore{As stated in the introduction,  if the representation $\rho$ is self dual up to a phase, that is, $\rho \simeq \rho^{*} \otimes \theta$ for some one-dimensional representation $\theta$, one may use the same representation $\rho$ for all systems. This condition holds in particular for all finite group representations with real valued characters, including all representations of the symmetric group.}
\end{remark}
\noindent 
The measurement and correction operations for Procedure~\ref{chargeteleportationopint}, together with the action $\sigma$ on the unspeakable classical channel, are \emph{correct data} if, regardless of Alice and Bob's reference frame alignments, Bob's system ends in the state  $\ket{\psi} \in \C^n$, according to a third observer with a fixed frame who can see both laboratories.
\begin{definition}[$G$-equivariant unitary error basis]
For a finite group $G$, and a Hilbert space $H$ carrying a unitary action $\rho$ of $G$, an \textit{equivariant unitary error basis} for $(G,\rho)$ is a unitary error basis $\{U_i\}_{i \in I}$ for $H$ whose elements are permuted up to a phase\footnote{In an early version of this work~\cite{Verdon2017} we used the term \emph{$G$-equivariant} for the specific  situation where $\xi(i,g)=1$. Here we choose to make this more general definition, since it is more physically relevant.
} by the right conjugation action of~$G$. 
\end{definition}

\noindent
That is,  for all $i \in I$ and $g \in G$, and some family of phases $\xi(i,g) \in \C$, we have that $\xi(i,g) \rho(g)^\dag U_i\rho(g) \in \{U_i\}_{i \in I}$.
Ignoring the phases, we can encode the effect of this conjugation as a right group action $\tau:  I \times G \to  I$.
\noindent
We now show that the notion of $G$\-equivariant unitary error basis gives a precise mathematical characterization of correct data for Procedure~\ref{chargeteleportationopint}.
\begin{theorem}\label{unspeakableteleportationconditions}
All correct data for Procedure~\ref{chargeteleportationopint} can be obtained from an equivariant unitary error basis $\{U_i\}$ for $(G,\rho)$, with associated right action $\tau$. The measurement and correction operations are as in~\eqref{eq:alicemeasbasis}, and the unspeakable classical channel carries the action $\tau^{-1}:G \times I \to I$. 
\ignore{ Bob's corrections are unitaries from an equivariant unitary error basis $\{U_i\}_{i \in I}$ for $(G,\rho)$, Alice measures in the basis~\eqref{eq:alicemeasbasis}, and the $G$-action on the unspeakable classical channel is inverse to that on the unitary corrections; that is,
\begin{equation}
\label{eqn:channelinversetooperationaction}
\; \forall\, g \in G,
\sigma(g,-) = \tau^{-1}(g,-)
\end{equation}}
\end{theorem}
\begin{proof}
We work in Alice's frame. Let Bob's misalignment with respect to this frame be $g \in G$. For sufficiency, suppose Alice measures $x \in I$; Bob then reads $\tau^{-1}(g,x)$ and performs the correction $$U_{\tau(\tau^{-1}(g,x),g)} = U_x,$$ as required. For necessity, note that the procedure must work for trivial misalignment $g = e$; therefore, by Werner's result~\cite[Theorem 1]{Werner2001}, Alice must perform  measurements corresponding to a unitary error basis, and Bob must perform the unitary correction $U_x$ in his own frame whenever he receives $x \in I$. The condition on the unspeakable channel is therefore clear.
\end{proof}
We say that an unspeakable classical channel is \emph{compatible} with an equivariant UEB when it carries the inverse action as in Theorem~\ref{unspeakableteleportationconditions}. We see that our scheme can be implemented for some representation $(G,\rho)$ if and only if there exists an equivariant UEB for $(G,\rho)$, and Alice and Bob have access to a compatible unspeakable classical channel. Before investigating these requirements, we draw a straightforward corollary from Theorem~\ref{unspeakableteleportationconditions}.

\begin{definition}[Orbit type]\label{orbittypedefinition}
For a $G$-equivariant unitary error basis $\{U_i\}_{i \in I}$, we  define its \emph{orbit type} as the multiset of sizes of each orbit in $I$ under the action \mbox{$\tau:I \times G \to I$}.
\end{definition}
\begin{corollary}\label{cor:speakableonly}
With only a speakable classical channel (that is, a channel carrying a trivial $G$-action), Procedure~\ref{chargeteleportationopint} succeeds for all frame alignments only if the action $\tau: I \times G \to I$ is trivial; that is,  the elements of the orbit type of the equivariant UEB are all $1$.
\end{corollary}
\noindent

\section{Unspeakable channels}\label{sec:unspchan}

In this section we address the physical requirement of our scheme, a compatible unspeakable classical channel for a given equivariant UEB. 

\subsection{Construction from quantum systems}
\label{sec:constructionfromquantumsystems}

We begin with a completely general method for constructing such a channel. When Alice performs the measurement on her two systems, they decohere in her measurement basis, and the joint system becomes a single classical object. Alice can transfer this directly to Bob, still in the eigenstate corresponding to her measurement result. Since the reference frame transformation is guaranteed to act as a permutation on measurement outcomes, Bob will also receive the system in an eigenstate, which he can can identify by performing the same measurement as Alice. Due to reference frame uncertainty, the result he receives may of course be different to that noted by Alice. The result is an unspeakable  classical channel. Since Bob both measures and performs the corresponding corrections in his own frame, the procedure will succeed for any reference frame misalignment.
\noindent

\subsection{Construction from shared classical system}
\label{sec:rfchannels}
In some physical situations, the method of Section~\ref{sec:constructionfromquantumsystems} involving transfer of the decohered quantum systems may be impractical. We now provide an alternative construction. The problem is the following: given the right action $\tau: I \times G \to I$ of a finite group on a finite index set, we must construct a compatible unspeakable classical channel $\Sigma$ whose set of messages $M_{\Sigma}$ can be identified with $I$, so that it carries the corresponding left action $\tau^{-1}: G \times I \to I$. 

Here we show how this can be done when $\tau^{-1}$ is a transitive action. This is sufficient since, if $\tau^{-1}$ is not transitive, $I$ will split into orbits under it, and the following procedure may be performed:
\begin{itemize}
\item After her measurement, Alice communicates the orbit $O\subset I$ of the index she measured, through a speakable channel.
\item She then communicates the precise measurement index $i \in O$ using an unspeakable classical channel with the set of messages $O$, carrying the restricted action $\tau^{-1}|_{O}:G \times O \to O$, which is transitive.
\end{itemize}
This procedure still leaks no reference frame information, since the orbit is communicated as speakable information and the outcomes within each orbit are equiprobable. It is still tight, since the classical channel distinguishes only $d^2$ possible messages, despite being split into speakable and unspeakable parts. It is still dynamically robust, since the orbit is unaffected by reference frame transformations.

We assume, therefore, that the action $\tau^{-1}$ is transitive. We can then characterise it further using the following well-known fact from group theory. Recall that the set of right cosets $\{Hg_i\}$ of a subgroup $H<G$ carries a canonical left action $ g\cdot (Hg_i) = Hg_i g^{-1}$; we write this left $G$-set as $G/H$.
\begin{lemma}
For any transitive left $G$-set $X$, there is a unique conjugacy class $C$ of subgroups of $G$ such that $X \simeq G/H$ iff $H \in C$.
\end{lemma} 
It follows that $\tau^{-1}$ is characterised up to isomorphism by its associated conjugacy class of subgroups. It also follows that any \emph{transitive} unspeakable classical channel $\Sigma$ (that is, any unspeakable classical channel whose set of messages $M_{\Sigma}$ is a transitive $G$-set) is characterised by its associated conjugacy class of subgroups $C_{\Sigma}$. Our problem can therefore be rephrased as follows: we need to construct a transitive unspeakable channel for which $C_{\Sigma} = C_{\tau^{-1}}$, so that $M_{\Sigma} \simeq G/H \simeq I$ as left $G$-sets. 

A key construction is the following, which allows us to group together messages in $M_{\Sigma}$ to create a new channel with a different associated conjugacy class.

\begin{construction}[Quotient channel]
Let $\Sigma$ be a transitive unspeakable classical channel with associated conjugacy class of subgroups $C_{\Sigma}$, and let $H_{\Sigma} \in C_{\Sigma}$. Fix an isomorphism $\alpha: M_{\Sigma} \simeq G/H_{\Sigma}$. Let $K$ be another subgroup such that $H_{\Sigma}< K <G$. 

We obtain a \emph{quotient channel} whose associated conjugacy class of subgroups has representative $K$, and whose messages are right cosets $K g$, transmitted as follows. In order to send a coset $K g$, Alice picks uniformly at random any element $x \in K/H_{\Sigma} \subset G/H_{\Sigma}$, and sends the message $\alpha^{-1}(xg) \in M_{\Sigma}$. Depending on his reference frame orientation, Bob receives some $y \in M_{\Sigma}$, such that $\alpha(y)$ lies in some right coset of $K/H_{\Sigma}$. He then uses the canonical isomorphism
$$
\frac{G/H_{\Sigma}}{K/H_{\Sigma}}
\simeq
G/K 
$$
to obtain a right coset of $K$ in $G$, which is the message he receives.
\end{construction}
\noindent
We obtain the following corollary. Recall the usual partial order on conjugacy classes of subgroups, where $C_1 < C_2$ iff $H_1 < H_2$ for some $H_1 \in C_1, H_2 \in C_2$.
\begin{corollary}\label{cor:canconstructcompat}
If we have access to a transitive unspeakable classical channel $\Sigma$  with associated conjugacy class of subgroups $C_{\Sigma}$, and $C_{\Sigma} < C_{\tau^{-1}}$, then we may construct a compatible channel for $\tau$.
\end{corollary}
\begin{proof}
Take $H_{\tau^{-1}} \in C_{\tau^{-1}}$, $H_{\Sigma} \in C_{\Sigma}$ such that $ H_{\Sigma} < H_{\tau^{-1}}$, and construct the quotient channel.
\end{proof}
\noindent
The trivial subgroup is the only member of its conjugacy class, which we call the \emph{trivial class}. The trivial class is the minimal element of the poset of conjugacy classes of subgroups. It follows that, from an transitive unspeakable channel $\Sigma$ whose associated conjugacy class of subgroups is the trivial class, we may construct a compatible channel for any transitive $\tau^{-1}$.

We now show how to use a shared classical system to construct an unspeakable classical channel with trivial associated conjugacy class.

\begin{definition}
A \emph{reference frame system} is a classical system  whose configuration is described according to a local reference frame, and whose set of configurations $C$ carries a free and transitive action of $G$.\end{definition}
\noindent 
The details of how this system is shared between Alice and Bob are abstracted away in this approach. The nomenclature is derived from the fact that Alice and Bob each possess physical systems serving as their local reference frames, on which the reference frame transformation group $G$ acts freely and transitively, by definition.

Alice and Bob will use their shared reference frame system to communicate messages. They associate each of the $|G|$ configurations of the system to an element of $G$ using a \emph{labelling}, which is a choice of isomorphism $l: C \to G$ depending on their local reference frame configurations. Once Alice fixes a labelling, she can communicate element $g \in G$ to Bob by preparing the system in the configuration associated to $g$ in her labelling. Bob will then interpret this configuration with respect to his own labelling. 

A labelling $l:C \to G$ is obtained by choosing a configuration $x_e$ such that $l(x_e) = e$;  the labelling is then fully determined by the equation $l (g\cdot x_e) = g l(x_e) = g$. Alice and Bob both agree on a way to pick $x_e$ based on their own local frame configuration; this is specified by a map $\epsilon: \mathcal{F} \to C$, where $\mathcal{F}$ is the space of local frame configurations and $\epsilon$ satisfies the naturality equation
$$
\epsilon(g\cdot f) = g \cdot \epsilon(f).
$$
We write $[l(x)]$ to refer to $x \in C$ when a labelling is fixed. Alice and Bob generally have different labellings $l_A$, $l_B$, so we write $[l_A(x)]_A$, $[l_B(x)]_B$ to refer to $x$ using their respective labellings. We obtain the following proposition.
\begin{proposition}\label{prop:readingtransfunderrfchange}A shared reference frame system gives rise to an transitive unspeakable classical channel whose associated conjugacy class of subgroups is trivial.
\end{proposition}
\begin{proof}
From the above discussion,the labelling of the channel is defined as $[g]_A = g \cdot [e]_A$; we have $[e]_A = \epsilon(f_A)$, so $[g]_A = g \cdot \epsilon(f_A) = g \cdot \epsilon(  g_{AB}^{-1} \cdot f_B) =(g g_{AB}^{-1}) \cdot [e]_B = [g g_{AB}^{-1}]_B$.
The channel therefore carries the action 
$
\sigma(g,x) = x g^{-1},
$
and the result follows.
\end{proof}
\noindent
By Corollary~\ref{cor:canconstructcompat}, it is therefore possible to construct a compatible unspeakable channel for any equivariant unitary error basis using a shared reference frame system. We conclude this section by presenting two examples of shared reference frame systems.

\begin{example}[Particle in a box]\label{ex:rfchannelbox}
Suppose that the quantum systems used in the teleportation protocol are particles in cubic boxes. In order to describe states of and operations on these systems, it is necessary to decide which sides of the box are `up', `front' and `right'. Alice and Bob shared such a labelling when they created their entangled pair of boxes; since that time, however, the orientation, and therefore the labelling, of Bob's box may have altered. The choice of labelling can be seen as a reference frame, whose transformation group is the group of rigid rotations of a cube. One reference frame system here is a classical solid cube, with labelled sides, passed between parties; the map $\epsilon: \mathcal{F} \to C$ is defined by labelling the cube identically to the box containing the particle. This is illustrated in Figure~\ref{fig:boxchannel}.
\begin{figure}
\includegraphics{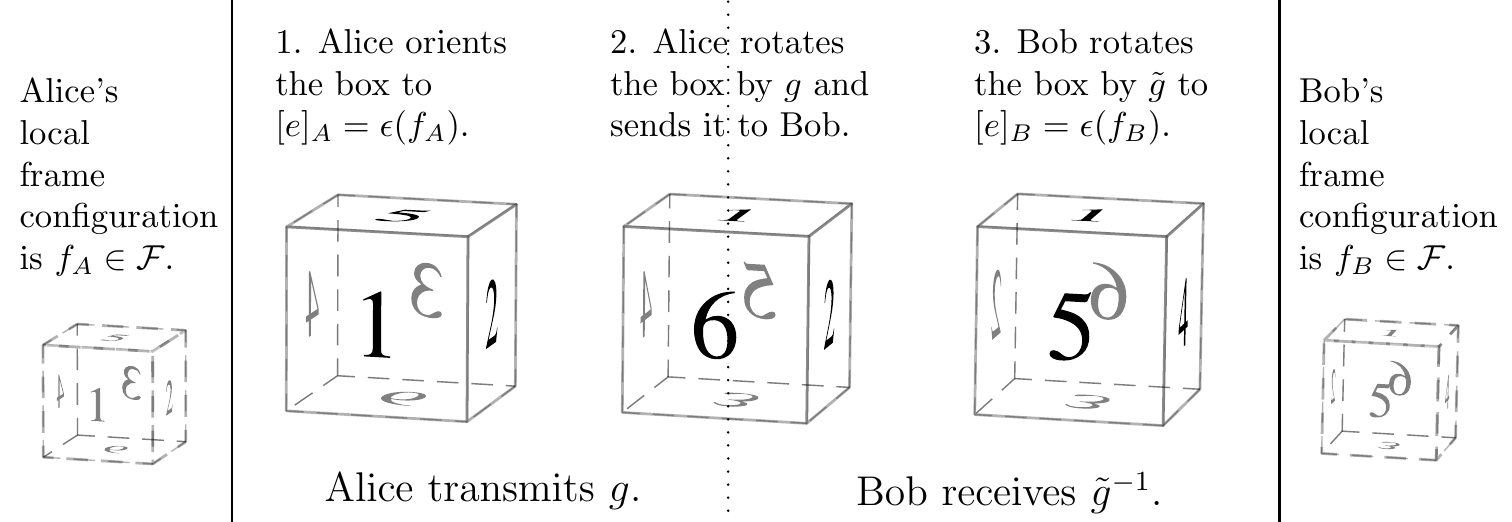}
\caption{The reference frame channel of Example~\ref{ex:rfchannelbox}, where $G$ is the group of rigid rotations of a cube. Here Alice transmits a $\pi/2$-rotation around the $x$ axis, and Bob receives a $\pi$-rotation around the $z$-axis.}
\label{fig:boxchannel}
\end{figure}
\end{example}

\begin{example}[Group of time translations]\label{ex:rfchanneltime} We suppose that the system to be teleported has a basis of energy eigenstates with different energy eigenvalues. Over the period $T$ of time evolution, these states will acquire a relative phase. In order to define states and operations, Alice and Bob must  choose a time $t_0$ at which the chosen basis vectors will have trivial phase. If we are promised that Alice and Bob's clocks are related by a time translation in a finite subgroup of $\U(1)$, then the choice of $t_0$ corresponds to a reference frame with cyclic transformation group. One reference frame system here is the time of arrival, modulo $T$, of a signal transmitted from Alice to Bob; the map $\epsilon: \mathcal{F} \to C$ is defined by the signal arriving at one's own time $t_0$. \end{example}

\section{Equivariant unitary error bases}
\label{sec:equivuebs}
We now turn to the classification and construction of equivariant unitary error bases, the mathematical basis for our scheme.

\subsection{Classification for qubits}\label{sec:qubitequebs}
We first fully classify equivariant UEBs for two-dimensional representations $(G, \rho)$. Let $q: \SU(2) \to \SO(3)$ be the quotient homomorphism taking a qubit unitary to its corresponding Bloch sphere rotation. Our results are outlined in the following theorem.
\begin{theorem}[Classification of equivariant UEBs for qubits]
\label{qubituebclassificationtheorem}%
The existence of unitary error bases of a given orbit type for a unitary representation $\rho: G \to \U(2)$ depends only on the isomorphism class of the image subgroup $q(\rho(G)) \subset \SO(3)$, according to the classification given in Table~\ref{qubitrepstable}.
\end{theorem}
\begin{table}[t]
\caption{UEB families for qubit representations. }
\label{qubitrepstable}
\centeringforarxiv{\version}
\begin{tabular}{lll}
\hline
\hline
\bf Isom. class of $q(\rho(G))$ & \bf Orbit types and solutions, up to phase & \bf Further details\\
\hline\hline
Trivial & (1,1,1,1) - any UEB & N/A
\\\hline
$\mathbb{Z}_2$ & (1,1,1,1) - one 2-parameter family & Proposition~\ref{z2equivariantrotations}
\\
 & (2,1,1) - one 2-parameter family &
 \\
 & (2,2) - one 2-parameter family &
 \\\hline
$\mathbb{Z}_3$ & (3,1) - one 2-parameter family & Proposition~\ref{Z3equivariantrotations}
\\\hline
$\mathbb{Z}_4$ & (2,1,1) - one 2-parameter family & Proposition~\ref{Z4equivariantrotations}
\\\hline
$\mathbb{Z}_n, n\geq5$ & No solutions & N/A
\\\hline
$D_2$ & (1,1,1,1) - one isolated solution & Proposition~\ref{d2equivariantuebs}
\\&(2,1,1) - six isolated solutions &
\\&(2,2) - three isolated solutions &
\\&(4) - two isolated solutions &
\\\hline
$D_3$ & (3,1) - six isolated solutions & Proposition~\ref{d3equivariantuebs}
\\\hline
$D_4$ & (2,1,1) - two isolated solutions & Proposition~\ref{d4equivariantuebs}
\\& (2,2) - two isolated solutions &
\\\hline
$D_n, n \geq 5$ & No solutions & N/A
\\\hline
Tetrahedral ($A_4$)& (4) - two isolated solutions & Proposition~\ref{tetrahedraluebs}
\\\hline
Octahedral ($S_4$) & (1,3) - one isolated solution & Proposition~\ref{octahedraluebs}
\\\hline
Icosahedral ($A_5$) & No solutions & N/A
\\\hline\hline
\end{tabular}
\end{table}
\noindent
Whilst in Table~\ref{qubitrepstable} we have only given the orbit type of the UEBs, in the proof we give now we  we also describe the associated action $\tau: I \times G \to G$. Before beginning the proof, we make a quick remark.
\begin{remark}\label{rem:qubitspeakabletel}
By Corollary~\ref{cor:speakableonly}, tight qubit teleportation without an unspeakable classical channel is possible only when the image of the composite homomorphism $G \stackrel \rho \to \U(2) \stackrel q \to \SO(3)$ is isomorphic to 1, $\mathbb{Z}_2$ or $D_2$. 
\end{remark}

We begin by fixing some notation for rotations. Euler showed~\cite{Euler1776} that every rotation in $\SO(3)$ can be represented uniquely as a rotation through an angle $0 \leq \theta \leq \pi$ around a given normalised vector $\hat{n} \in \mathbb{R}^3$. We write a rotation through an angle $\theta$ around an axis $\hat{n}$ as $r(\theta,\hat{n})$.\footnote{Note that this notation is slightly redundant because rotations through an angle $\pi$ around antipodal $\hat{n}$ are identical, as are all rotations through an angle $0$.} Given two rotations $r(\theta_1,\hat{n}_1)$ and $r(\theta_2,\hat{n}_2)$, we write the angle and axis of the composite as $\theta_{12}$ and $\hat{n}_{12}$. For concision, we will occasionally write rotations simply as $r \in \SO(3)$, omitting to mention the axis and angle of rotation.

It is well known that unitary operations on a qubit correspond to rotations of the Bloch sphere together with a global phase~\cite[Exercise 4.8]{Nielsen2011}. It is easy to check that two unitaries $U_1,U_2$ are orthogonal iff their corresponding Bloch sphere rotations $q(U_1),q(U_2)$ are orthogonal in the following sense.
\begin{definition}\label{def:orthogrotations}
Two rotations $r_1,r_2 \in \SO(3)$ are \emph{orthogonal} if the composite $r_1^{-1} r_2$ is a rotation through the angle $\pi$.
\end{definition}
\noindent
The image of a UEB under the quotient $q$ will be a set of orthogonal rotations preserved under conjugation by the orthogonal rotations $q(\rho(g))$ for $g \in G$; this inspires the following definition.
\begin{definition}
\label{def:oeb}
An \textit{orthogonal error basis} (OEB) is a family $\mathcal O \subset \SO(n)$ of $n^2$ orthogonal rotations. For a finite group $G$ and a homomorphism $\rho: G \to \SO(n)$, an \textit{equivariant orthogonal error basis} for $(G,\rho)$ is an OEB $\mathcal O \subset \SO(n)$  preserved under conjugation by $\rho(g)$ for all $g\in G$.
\end{definition}
\ignore{An equivariant unitary error basis for $(G,\rho)$ will be mapped under $q$ to an equivariant orthogonal error basis for $(G, q \circ \rho)$; }
In the other direction, given an equivariant OEB for $(G,q \circ \rho)$, one may obtain all corresponding equivariant UEBs for $(G,\rho)$ by picking phases for each rotation.
\ignore{
\begin{proposition}\label{uebexistsiffprojexiststhm}
Let $\pi:G \to \U(2)$ be a representation. Then a $G$-equivariant UEB exists for $\pi$ if and only if a $G$-equivariant OEB exists for $q \circ \pi$.
\end{proposition}
\noindent
}
A classification of equivariant UEBs for subgroups $G \subset \U(2)$ is therefore equivalent to a classification of equivariant OEBs for subgroups $q(G) \subset \SO(3)$.
Note also that the action of $\rho(g)$ on the index set of a UEB is identical to the action of $q(\rho(g))$ on the index set of the corresponding OEB.
\begin{theorem}[{\cite[Theorem 19.2]{Armstrong1997}}]
The finite subgroups of $\SO(3)$ are as follows:
\begin{itemize}
\item cyclic groups $\Z_n$ for $n \geq 1$, generated by a rotation through $2\pi/n$ around a given axis;
\item dihedral groups $D_n$ for $n \geq 1$, generated by a rotation through $2\pi/n$ around a given axis and a $\pi$-rotation around a perpendicular axis;
\item the group of orientation-preserving symmetries of a regular tetrahedron, isomorphic to $A_4$;
\item the group of orientation-preserving symmetries of a regular octahedron (or a cube), isomorphic to $S_4$;
\item the group of orientation-preserving symmetries of a regular icosahedron, isomorphic to $A_5$.
\end{itemize}
\end{theorem}
\noindent
In order to find sets of points preserved under the conjugation action of these subgroups, we recall a useful way to think about conjugation in $\SO(3)$. The group $\SO(3)$ may be viewed as a closed ball $\B(3) \subset \mathbb{R}^3$ of radius $\pi$, which we call the $\SO(3)$-ball, under the identification 
\begin{equation}\label{eq:so3ballidentification}
r(\theta, \hat{n}) \mapsto \theta \hat{n}.
\end{equation}
Antipodal points on the boundary are identified, since rotation through an angle $\pi$ around $\hat{n}$ is the same as rotation through an angle $\pi$ around $-\hat{n}$. Given two rotations $r_1 = r(\theta, \hat{n})$ and $r_2$, we have the identity $$r_2 r_1 r_2^{-1} = r_2 r(\theta,\hat{n}) r_2^{-1} =r(\theta,r_2(\hat{n})).$$ It follows that, under the identification~\eqref{eq:so3ballidentification}, conjugation by a rotation in $\SO(3)$ corresponds to rotation of the $\SO(3)$-ball. Equivariant OEBs for a subgroup are therefore sets of orthogonal points in the $\SO(3)$-ball permuted by rotations in that subgroup.

For concision, in what follows we will occasionally conflate points in $\B(3)$ and rotations in $\SO(3)$. For instance, we say `a point on the $z$-axis' to signify the element of $\SO(3)$ corresponding to a point on the $z$-axis, that is, a rotation around the $z$-axis through some angle. We will also write $\mathrm{sin}(x)$, $\mathrm{cos}(x)$ and $\mathrm{tan}(x)$ as $\sin(x)$, $\cos(x)$ and $\tan(x)$ respectively.

We now recall some useful facts about orthogonality in $\SO(3)$. 
\begin{lemma}\label{rotationslemma1}
Each rotation in $\SO(3)$ around $\hat{n}$ is orthogonal to exactly one other rotation around $\pm \hat{n}$.
\end{lemma}
\begin{proof}
The composite $r(\theta_1, \hat{n})^{-1} r(\theta_2, \hat{n})$ is the rotation $r(\theta_2 -\theta_1, \hat{n})$. For a given $\theta_1 \in [0,\pi]$, there is only one $\theta_2 \in (-\pi,\pi]$ such that $\theta_1 - \theta_2$ is an odd multiple of $\pi$.
\end{proof}

\begin{lemma}\label{rotationslemma2}
The rotation $r(\theta_2, \hat{n}_2)$ is orthogonal to the rotation $r(\pi,\hat{n}_1)$ iff either $\hat{n}_2$ is orthogonal to $\hat{n}_1$ or $\theta_2 = 0$. 
\end{lemma}
\begin{proof}
We have the following standard formula for the rotation angle $\theta_{12}$ of the composite $r_2^{-1} \circ r_1$, where $r_i$ is a rotation around the axis $\hat{n}_i$ through an angle $\theta_i \in [0,\pi]$~\cite[Exercise 4.15]{Nielsen2011}:
\begin{align}
\begin{split}
\cos(\theta_{12}/2) = &\cos(\theta_1/2) \cos(\theta_2/2) \\&+\sin(\theta_1/2) \sin(\theta_2/2)\hat{n}_1 \cdot \hat{n}_2\label{angleofcompositeeqn}
\end{split}
\end{align}
Orthogonality of $r_2$ and $r_1$ is precisely the condition that the LHS is zero. Since the first term on the RHS equals zero when $\theta_1 = \pi$, the second term must also. This implies that either $\hat{n}_1 \cdot \hat{n}_2 = 0$, in which case the axes of rotation are orthogonal, or $\sin(\theta_2/2)=0$, in which case the other rotation is simply the identity.
\end{proof}

\begin{lemma}\label{rotationslemma3}
Two rotations can be orthogonal only if the angle between the axes of rotation is obtuse. If the angle between the axes is $\pi/2$ then for orthogonality one rotation must be through the angle $\pi$.
\end{lemma}
\begin{proof}
Considering (\ref{angleofcompositeeqn}), we note that both $\cos(\theta_1/2) \cos(\theta_2/2)$ and $\sin(\theta_1/2) \sin(\theta_2/2)$ will be positive for $\theta_1,\theta_2 \in [0,\pi]$. The sum can only be zero, then, if $\hat{n}_1\cdot\hat{n}_2 \leq 0$, i.e. if the angle between the axes is obtuse. If the angle is $\pi/2$ then we need $\cos(\theta_1/2) \cos(\theta_2/2)= 0$, which implies that one of the rotations is through an angle~$\pi$.
\end{proof}

We now begin our classification.

\subsubsection{Cyclic subgroups of $\SO(3)$}

Any set of orthogonal points will be equivariant for $\mathbb{Z}_1$. We proceed directly to the nontrivial cases. Let the $z$-axis be the axis of rotation of the generator of $\mathbb{Z}_n$ which rotates the $\SO(3)$-ball through an angle $2\pi/n$. Recalling that antipodal points on the ball's surface are identified, we immediately obtain the following characterisation of the orbits under this action.
\begin{lemma}\label{cyclicorbitsizelemma}
The orbit sizes under the conjugation action of $\mathbb{Z}_n$ on $\SO(3)$ are: 
\begin{itemize}
\item 1, for a point on the axis of rotation; 
\item
$n$, for a point in the interior of the ball and not on the axis of rotation, on the boundary of the ball and not on the $xy$-plane or the axis of rotation, or on the intersection of the boundary of the ball and the $xy$-plane when $n$ is odd; 
\item $n/2$, for a point on on the intersection of the boundary of the ball and the $xy$-plane when $n$ is even.
\end{itemize}
\end{lemma}\noindent
\begin{proposition}\label{z2equivariantrotations}
The $\mathbb{Z}_2$-equivariant orthogonal error bases are as follows:
\begin{itemize}
\item for orbit type (1,1,1,1), a 2-parameter family of solutions, where two points are rotations around the $z$-axis and the other two are $\pi$-rotations around orthogonal axes in the $xy$-plane;
\item for orbit type (2,1,1), a 2-parameter family of solutions, where one point is a rotation around the $z$-axis, another point is a $\pi$-rotation around an $x$-axis perpendicular to the $z$-axis, and the other two points are rotations around axes in the $yz$-plane (see Figure~\ref{fig:z221orbits}), where the $y$-axis is perpendicular to both the $x$- and $z$-axes;
\item for orbit type (2,2), a 2-parameter family of solutions, where, for an axis $x$ orthogonal to $z$ and an axis $y$ orthogonal to both, two points lie in the $xz$-plane and below the $xy$-plane, and another two points lie in the $yz$-plane and above the $xy$-plane (see Figure~\ref{fig:z222orbits}).
\end{itemize}
\end{proposition}

\begin{proof}
\emph{Orbit type (1,1,1,1)}. By Lemma~\ref{rotationslemma1} there can be at most two rotations on the $z$-axis. The other two, in order to have orbit size 1, must both be $\pi$ rotations around different axes in the $xy$-plane, which must be orthogonal to each other by Lemma~\ref{rotationslemma2}. This set of solutions therefore has two independent parameters, namely the angle of one rotation around the $z$-axis and the orientation of the perpendicular axes in the $xy$-plane.

\item \emph{Orbit type (2,1,1)}. Firstly, suppose both the 1-orbits lie off the $z$-axis. Then they must be orthogonal $\pi$-rotations in the $xy$-plane. But then the other two rotations would have to be orthogonal and we would end up in the case $(1,1,1,1)$.

Let us now suppose that exactly one of the $1$-orbits lies on the $z$-axis. The other must be an orthogonal $\pi$-rotation; let this be around the $x$\-axis. Then the $2$-orbit must lie in the $yz$-plane by Lemma \ref{rotationslemma2}. We are therefore looking for three orthogonal points in the $yz$-plane, one on the $z$-axis and the other two symmetric under a reflection in the $z$-axis. Let $r$ be the rotation angle of the elements in the $2$-orbit and $\theta$ be the angle between them. Here we take $0 <\theta < 2\pi$, where $\theta = 0$ would correspond to both points being on the positive $z$-axis. By (\ref{angleofcompositeeqn}) we have the following equation for orthogonality of the elements of the 2-orbit:
\begin{equation}\label{eq:r211z1=2orbit}
r = 2\cos^{-1}\left(\sqrt{\frac{\cos(\theta)}{\cos(\theta)-1}}\right)
\end{equation}
This has a unique solution $r\in [\pi/2,\pi]$ for $\theta \in [\pi/2,3\pi/2]$, and none otherwise. Using (\ref{angleofcompositeeqn}), it can be shown similarly that, for given $\theta$, there is a unique value of the $z$-coordinate of the 1-orbit such that all three points are orthogonal (see Figure~\ref{fig:z221orbits}). We therefore have a 2-parameter family of solutions, where one parameter corresponds to a choice of $z$-coordinate $z_1$ of the 1-orbit on the $z$-axis, and the other parameter comes from a choice of orientation of $x$-axis. 
\begin{figure}
\centering
\includegraphics{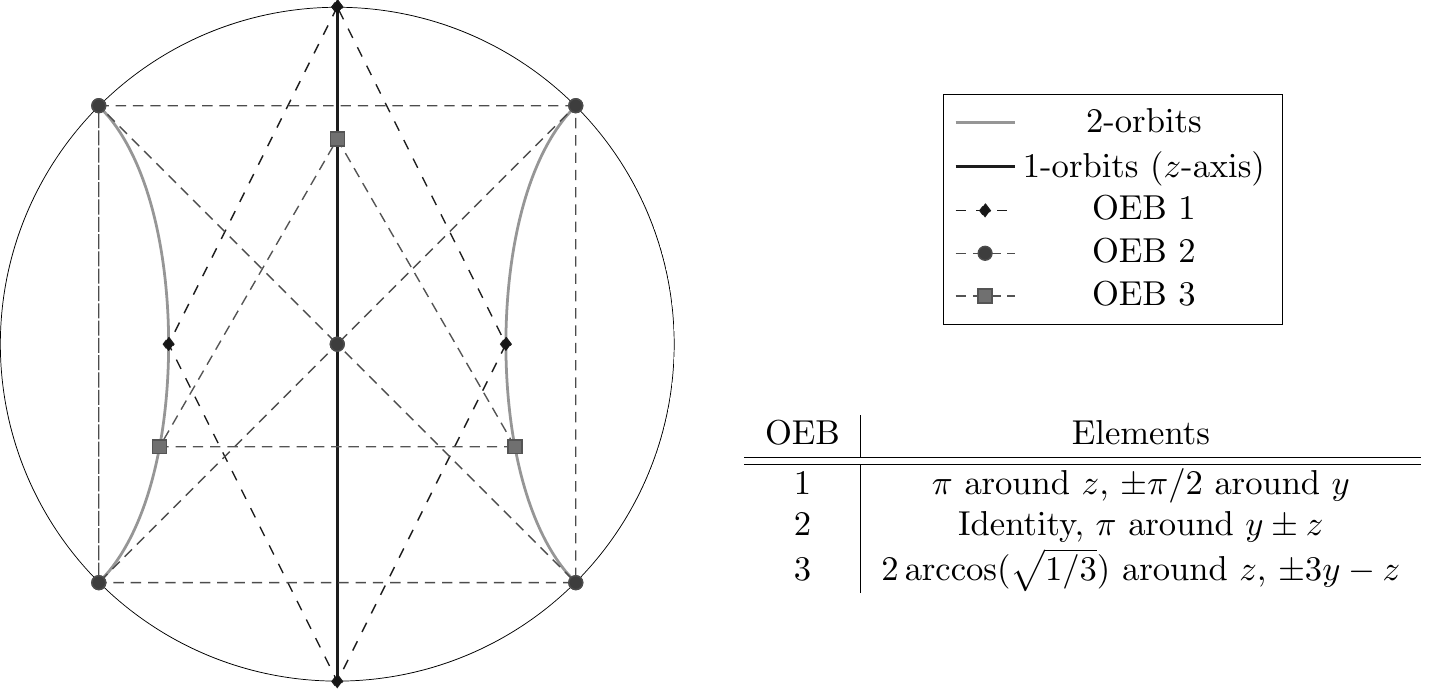}
\caption{\small $\mathbb{Z}_2$-equivariant OEBs with orbit type $(2,1,1)$. The diagram shows the intersection of the $yz$-plane with the $\SO(3)$-ball. One 1-orbit of the OEB is a $\pi$-rotation around the $x$-axis, and the remaining 2-orbit and 1-orbit are rotations around axes in the $yz$-plane shown in the diagram. Each 2-orbit is a pair of points with identical $z$-value on the two curved gray lines. The corresponding 1-orbit is a point on the $z$-axis. Three possible choices of points are given in the table and marked in the figure, joined by dashed lines.}
\label{fig:z221orbits}
\end{figure}

Suppose now that both 1-orbits lie on the $z$-axis; we will demonstrate that we cannot then obtain solutions of this orbit type. Firstly, if the elements of the 2-orbit are $\pi$-rotations not in the $xy$-plane, then they will not be orthogonal to the 1-orbits on the $z$-axis. On the other hand, if the elements of the 2-orbit are rotations through an angle less than $\pi$ and not in the $xy$-plane, then, given that by Lemma \ref{rotationslemma1} the $z$-rotations will be on opposite sides of the origin, both elements of the 2-orbit will make an acute angle with one of the $z$-rotations, violating Lemma \ref{rotationslemma3}. The 2-orbit must therefore lie in the $xy$-plane. The rotations of the 2-orbit must be through an angle less than $\pi$, or they would form two $1$-orbits. But, by Lemma \ref{rotationslemma3}, in order to be orthogonal both $z$-rotations must then be through an angle $\pi$, which would identify them.

\item \emph{Orbit type (2,2)}. Each 2-orbit will lie in a plane through the $z$-axis. Again, let $r$ be the rotation angle of the elements in the 2-orbit and $\theta$ be the angle between them; the relationship between $r$ and $\theta$ was already given in~\eqref{eq:r211z1=2orbit}.

We must find two 2-orbits where all four elements are pairwise orthogonal. Without loss of generality let the first orbit $O_1$ lie in the $xz$-plane, and let $\theta_1 \in [\pi/2,\pi]$. Certainly, the second orbit $O_2$ must have $\theta_2 \in [\pi,3\pi/2]$, as otherwise the central angle between some pair of elements will be acute. We now show that the orbit $O_2$ must also lie in the $yz$-plane. In other words, the two 2-orbits must lie in orthogonal planes containing the $z$-axis, and be on opposite sides of the $xy$-plane.

Let $r_1, r_2 \in [0, \pi]$ be the rotation angles of $O_1$ and $O_2$ respectively. Take one element from each orbit, and consider their composition~\eqref{angleofcompositeeqn}. With $r_1,r_2$ fixed, the only thing that can vary on the right hand side of this equation is the angle between the axes of rotation of these elements. This angle will lie between $0$ and $\pi$, and $\cos(x)$ is single-valued in that range; therefore, for both elements of the second orbit to be orthogonal to the given element of the first, their axes of rotation must both have an equal central angle with that element. This means that the $xz$-plane containing $O_1$ must be orthogonal to the plane through the $z$-axis containing $O_2$, which must therefore be the $yz$-plane.

With the planes fixed, we now find which angles $\theta_1 \in [\pi/2,\pi]$ and $\theta_2 \in [\pi,3\pi/2]$ defining the two orbits are compatible. By the above discussion, for orthogonality of all elements it is sufficient for a single pair of elements from different orbits to be orthogonal. Unit vectors $\hat{n}_1, \hat{n}_2$ defining the axes of rotation of a pair of elements in $O_1, O_2$ respectively may be expressed in Cartesian coordinates as $\hat{n}_1 = (\sin(\theta_1/2)),0,\cos(\theta_1/2))$ and $\hat{n}_2 = (0, \sin(\theta_2/2),\cos(\theta_2/2)).$ The orthogonality condition~\eqref{angleofcompositeeqn} then becomes
\begin{align}\label{z222orthogonalityequation}
\begin{split}
-\cos(r_1/2)\cos(r_2/2) = \sin(r_1/2) \sin(r_2/2)\cos(\theta_1/2)\cos(\theta_2/2)
\end{split}.
\end{align}
Replacing $\theta_1, \theta_2$ with $r_1, r_2$ using~\eqref{eq:r211z1=2orbit}, squaring both sides and performing some trigonometric manipulations, we derive
$$
r_1 = 2 \cos^{-1}\left(\sqrt{\frac{1}{2}-\cos^2(\frac{r_2}{2})}\right)
$$
\ignore{$$
\cos(\theta_1) = -\frac{1+\cos(\theta_2)}{1+3\cos(\theta_2)}.
$$}
This uniquely determines \mbox{$r_1 \in [\pi/2,\pi]$} for any \mbox{$r_2 \in [\pi/2,\pi]$}. The solutions of orbit type (2,2) are therefore parametrised by two angle variables; the first is the orientation of the $x$-axis and the second is the angle $r_2$ of one of the rotations in the 2-orbit $O_2$ lying below the $xy$-plane.
Two of these solutions are shown in Figure~\ref{fig:z222orbits}.
\begin{figure}
\includegraphics{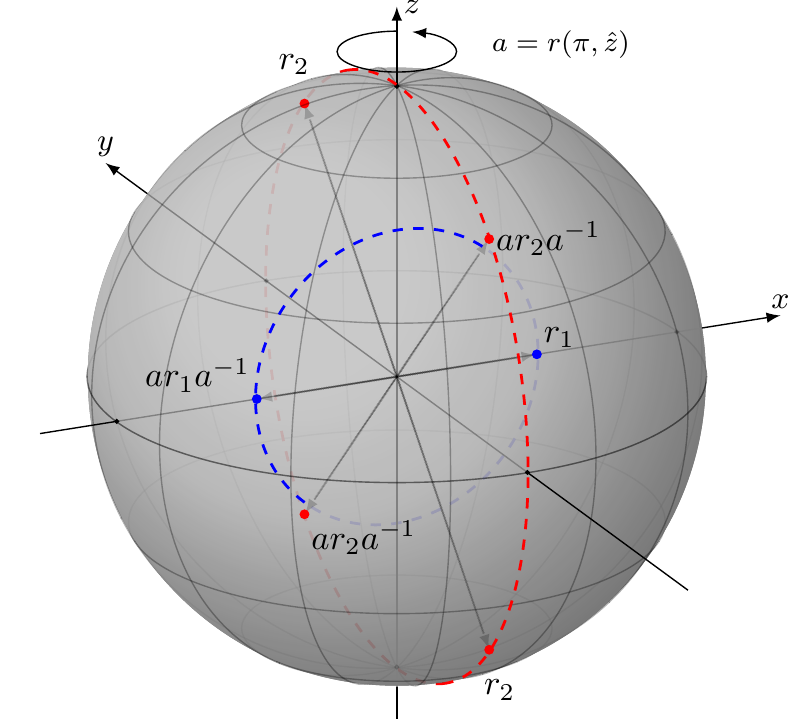}
\quad
\includegraphics{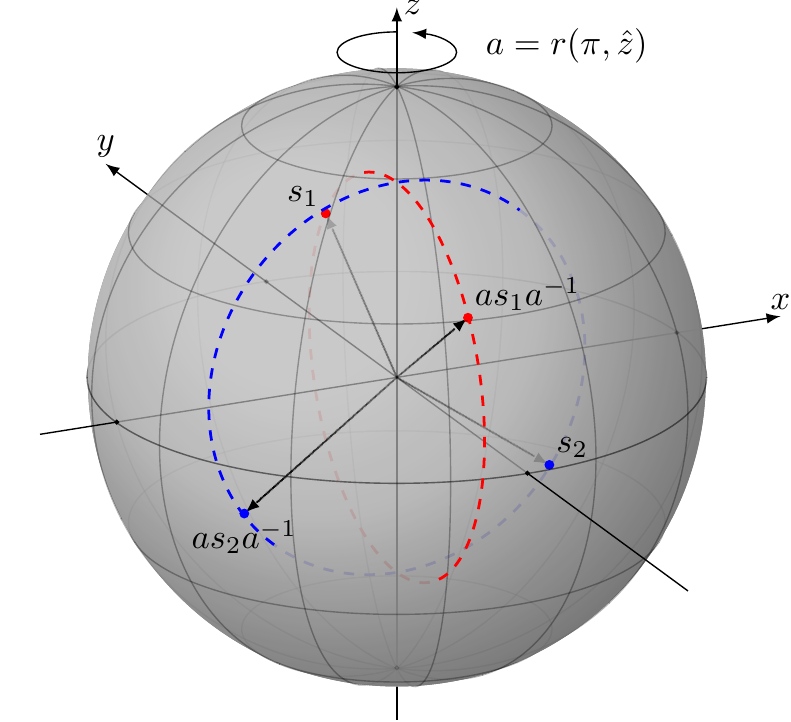}
\caption{Two equivariant OEBs for $\mathbb{Z}_2$ with orbit type (2,2), pictured in the $\SO(3)$-ball. Under the $\mathbb{Z}_2$ action, the equivariant OEB on the left is generated by $r_1 = r(\pi/2, \hat{x})$ and $r_2 = r(\pi, \frac{1}{\sqrt{2}}(\hat{y} + \hat{z}))$ (note the identification of antipodal points), while the equivariant OEB on the right is generated by $s_1 = r(2\pi/3, \frac{1}{\sqrt{3}}(\sqrt{2}\hat{y}+ \hat{z})) $ and $s_2= r(2\pi/3, \frac{1}{\sqrt{3}}(\sqrt{2}\hat{x}- \hat{z})).$}
\label{fig:z222orbits}
\end{figure}
\end{proof}

\begin{proposition}\label{Z3equivariantrotations}
The $\mathbb{Z}_3$-equivariant  orthogonal error bases are as follows:
\begin{itemize}
\item for orbit type (1,1,1,1), no solutions;
\item for orbit type (3,1), a 2-parameter family of solutions, forming the vertices of a tetrahedron with one vertex on the $z$-axis and the other three forming an equilateral triangle in a plane perpendicular to the $z$-axis (see Figure~\ref{fig:z331orbits}).
\end{itemize}
\end{proposition}

\begin{proof} \emph{Orbit type (1,1,1,1).} All the points would need to be on the $z$-axis, which is impossible by Lemma \ref{rotationslemma1}.
\item \emph{Orbit type (3,1).} By the classification of orbits (Lemma~\ref{cyclicorbitsizelemma}), these OEBs consist of a 1-orbit on the $z$-axis and a 3-orbit forming the vertices of an equilateral triangle in a plane perpendicular to the $z$-axis. Let one of the elements in the $3$-orbit lie in the $xz$-plane, so $(r,\psi,0)$ are its spherical coordinates. From \eqref{angleofcompositeeqn} we obtain the following condition for orthogonality of the elements of the  3-orbit: 
\begin{align*}
r &= 2 \sin^{-1}\left(\frac{\sqrt{2}}{\sqrt{3}\sin(\psi)}\right)
\end{align*}
Where soluble, this equation completely determines $r$ for given $\psi$. It admits solutions for $\psi \in [\sin^{-1}(\sqrt{\frac{2}{3}}), \pi-\sin^{-1}(\sqrt{\frac{2}{3}})]$. \ignore{\DVcomm{solutions $(r,\psi)$ of interest are $(2s^{-1}(\sqrt{\frac{2}{3}}),\frac{\pi}{2})$, $(\pi,s^{-1}(\sqrt{\frac{2}{3}}))$ and $(\pi, \pi-s^{-1}(\sqrt{\frac{2}{3}}))$.}} By~\eqref{angleofcompositeeqn} we also obtain an equation in $\psi$ for the height $z$ of the point on the $z$-axis, which is single-valued in the range $\psi \in [\sin^{-1}(\sqrt{\frac{2}{3}}), \pi-\sin^{-1}(\sqrt{\frac{2}{3}})]$:
$$
z=2\tan^{-1}(\sqrt{\frac{3}{2}}\cos(r(\psi)/2)\tan(\psi))
$$
Under this equation $z$ can take all values in $[-\pi,\pi]$; the 3-orbit lies on the other side of the $xy$-plane.  These OEBs therefore form a 2-parameter family, where one parameter is the angle $\psi$, and the other is the choice of $x$-axis. Two solutions are shown in Figure~\ref{fig:z331orbits}.
\begin{figure}
\includegraphics{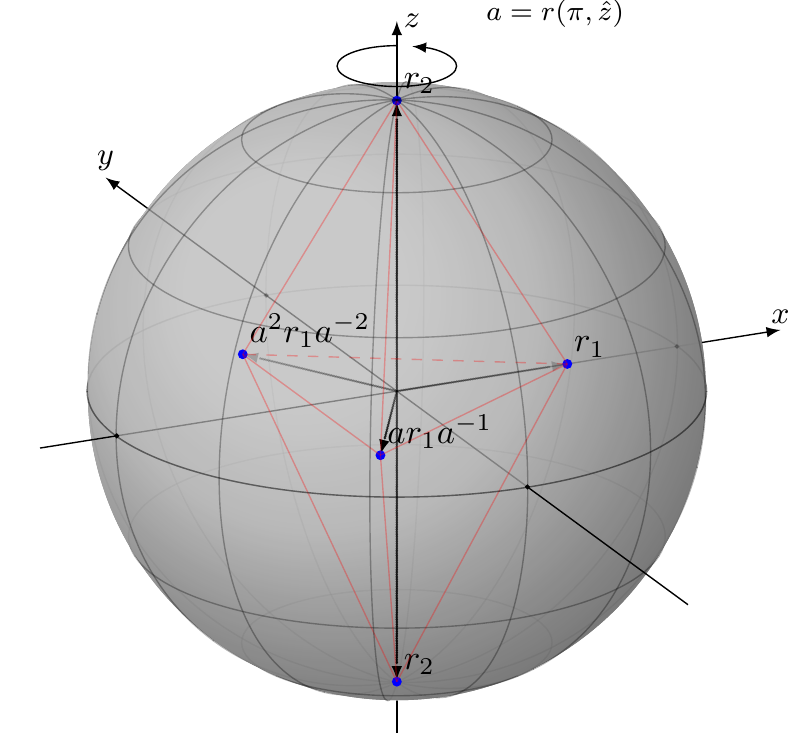}
\quad
\includegraphics{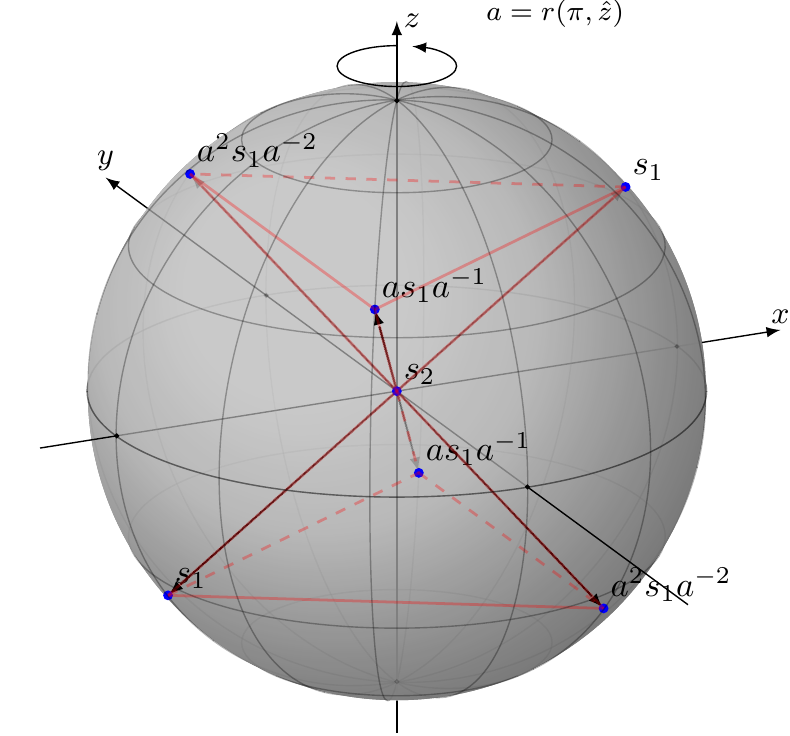}\caption{Two equivariant OEBs for $\mathbb{Z}_3$ with orbit type (3,1). Under the $\mathbb{Z}_3$ action, the equivariant OEB on the left is generated by  $r_1 = r(2\sin^{-1}(\sqrt{\frac{2}{3}}),\hat{x})$ and $r_2 = r(\pi, \hat{z})$, and the equivariant OEB on the right is generated by  $s_1 = r(\pi, \frac{1}{\sqrt{3}} (\sqrt{2}\hat{x}+\hat{z}))$ and $s_2 = r(0, \hat{z})$. Note the identification of antipodal points in both cases; this is why the points are vertices of two tetrahedra rather than just one.}
\label{fig:z331orbits}
\end{figure}
\end{proof}

\begin{proposition}\label{Z4equivariantrotations}
The $\mathbb{Z}_4$-equivariant  orthogonal error bases are as follows:
\begin{itemize}
\item for orbit type (1,1,1,1), no solutions;
\item for orbit type (2,1,1), a 2-parameter family of 
solutions identical to the (1,1,1,1) solutions for $\mathbb{Z}_2$ (Proposition~\ref{z2equivariantrotations});
\item for orbit type (2,2), no solutions;
\item for orbit type (4), no solutions.
\end{itemize}
\end{proposition}

\begin{proof}
\emph{Orbit type (1,1,1,1).} All the points would need to be on the $z$-axis, which is impossible by Lemma \ref{rotationslemma1}.

\item\emph{Orbit type (2,1,1).} The 2-orbit must consist of orthogonal $\pi$-rotations around axes in the $xy$-plane. One parameter therefore corresponds to the rotation angle of one of the rotations on the $z$-axis, and the other to the orientation of the orthogonal axes in the the $xy$-plane.

\item\emph{Orbit type (2,2).} These must be four different $\pi$-rotations around axes in the $xy$-plane. But then they cannot be orthogonal.

\item\emph{Orbit type (4).} The angle between rotation vectors in a 4-orbit will be acute if they are not in the $xy$-plane, so they cannot be orthogonal. If they are in the $xy$-plane then as the angle between adjacent vectors is $\pi/2$, at least one pair of opposite vectors must be $\pi$-rotations by Lemma \ref{rotationslemma3}; but then these will be identified and this will not be a 4-orbit.
\end{proof}

\begin{proposition} \label{cyclicqubituebs}
There are no $\mathbb{Z}_n$-equivariant orthogonal error bases for $n \geq 5$. 
\end{proposition}

\begin{proof}
We handle the odd and even cases separately.

\item \emph{$n \geq 5$ and $n$ odd}: The only orbit sizes are $1$ and $n$. Since we only have four elements in the UEB, all four points must be of orbit size $1$; they must therefore all be on the $\hat{z}$-axis. But this is impossible by Lemma~\ref{rotationslemma1}.

\item \emph{$n\geq 5$ and $n$ even}: For $n=6$, the orbit sizes are $1$, $3$ and $6$. Since for the reason given above we cannot have four $1$-orbits, we must have one $1$-orbit and one $3$-orbit. However, the axes of the $\pi$-rotations will not be orthogonal and so the rotations are not orthogonal by Lemma \ref{rotationslemma2}. For $n=8$, the orbit sizes are $1, 4$ and $8$, so we are forced to have a 4-orbit by Lemma \ref{rotationslemma1}. But these $\pi$ rotations will again not be around orthogonal vectors and are therefore not orthogonal by Lemma \ref{rotationslemma2}. For $n>8$, the orbit sizes for elements off the $\hat{z}$-axis are all greater than~4.
\end{proof}

\subsubsection{Dihedral subgroups of $\SO(3)$}

Let the $z$-axis be the axis of cyclic rotation, and let the $f$-axis be the perpendicular axis of $\pi$-rotation (the `flip axis'). 

\begin{proposition}\label{d2equivariantuebs}
The $D_2$-equivariant orthogonal error bases are as follows:
\begin{itemize}
\item for orbit type (1,1,1,1), one solution;
\item for orbit type (2,1,1), six solutions;
\item for orbit type (2,2), three solutions;
\item for orbit type (4), two solutions.
\end{itemize}
\end{proposition}

\begin{proof}
Any solution for $D_2$ must also be a solution for its $\mathbb{Z}_2$ subgroup, and we proceed by case analysis of $\mathbb{Z}_2$-orbit types, making use of Proposition~\ref{z2equivariantrotations}.

\item \emph{$\mathbb{Z}_2$-orbit type $(1,1,1,1)$.} Recall that $\mathbb{Z}_2$-equivariant OEBs of this type are made up of two $\pi$-rotations around orthogonal axes in the $xy$-plane and two rotations around the $z$-axis. If we fix the flip axis $f$, in order that the rotations in the $xy$-plane are preserved there are two choices for the axes; either $f$ and $g$, or $f+g$ and $f-g$. In order that the $z$-rotations are preserved, there are two choices for the rotation angles; either $0$ and $\pi$, or $-\pi/2$ and $\pi/2$. The orbit types are (1,1,1,1), (2,1,1), (2,1,1) and  (2,2).

\item\emph{$\mathbb{Z}_2$-orbit type (2,1,1).} Recall that $\mathbb{Z}_2$-equivariant OEBs of this type are made up of a $\pi$-rotation around some $x$-axis, a rotation around the $z$-axis, and two other rotations around axes in the $yz$-plane (see Figure~\ref{fig:z221orbits}).  Fix the flip axis $f$. The $z$-rotation will be preserved under the flip only if it is through an angle $\pi$ or $0$. This fixes the rotation angle $r$ of the elements in the $2$-orbit as $\pi/2$ or $\pi$ respectively. For the $x$-rotation to be preserved under the flip, we need either that $x=f$ or $y=f$. In both of these cases, the solutions with $r=\pi/2$ and $r=\pi$ are preserved. We therefore obtain four equivariant OEBs with orbit type (2,1,1).

\item \emph{$\mathbb{Z}_2$-orbit type (2,2).} Consider the 2-parameter family of solutions of orbit type $(2,2)$. The 2-orbits $O_1$, $O_2$ lie on opposite sides of the $xy$-plane, in the $xz$- and $yz$-planes respectively. 
$D_2$ is abelian, so the $\mathbb{Z}_2$-orbit pairing will be preserved after the flip. There are therefore two possibilities if the elements are to be preserved under the flip; the flip can either swap the $xz$- and $yz$-planes or preserve them. 

If the planes are preserved then the flip axis must be the $x$- or $y$-axis, and the $2$-orbits must be symmetric under reflection in the $xy$-plane. Since one orbit is fixed by the other, this gives two solutions of orbit type $(2,2)$, corresponding to a choice of $r_1=\pi/2$ or $r_1=\pi$ in $O_1$, where $r_i$ is the rotation angle of the elements of $O_i$ (see Figure~\ref{fig:z221orbits}).

If the planes are permuted then the flip axis must be  $v_1 \pm v_2$, and $r_1=r_2$. Setting $r_1=r_2$ in~\eqref{z222orthogonalityequation} and substituting in~\eqref{eq:r211z1=2orbit}, we obtain $\cos(\theta) = -\frac{1}{3}$, where $\theta \in [\pi/2,\pi]$ is the central angle between the elements of each orbit.
This has a unique solution in the relevant domain, of orbit type (4). There are two of these for a given choice of $f$-axis, since we can choose which orbit lies above the $xy$-plane.
\end{proof}

\begin{proposition}
\label{d3equivariantuebs}
There are six isolated $D_3$-equivariant quotient orthogonal error bases all of orbit type~(3,1).
\end{proposition}

\begin{proof}
Any solution for $D_3$ must also be a solution for its $\mathbb{Z}_3$ subgroup. In Proposition~\ref{Z3equivariantrotations} we saw that solutions were the vertices of a 2-parameter family of tetrahedra with one vertex on the $z$-axis and the others forming the vertices of an equilateral triangle on the other side of the $xy$-plane. The vertex on the $z$-axis point must be preserved under reflection in the $xy$-plane, so it must be through an angle $0$ or $\pi$; the two possibilities were shown in Figure~\ref{fig:z331orbits}. For $z=0$, the elements of the $3$-orbit will be preserved if the $fz$ plane is orthogonal to the triangle's medians, giving three solutions. For $z=\pi$, the $f$-axis must go through any of the three vertices of the triangle, giving three solutions.
\end{proof}

\begin{proposition}\label{d4equivariantuebs}
The $D_4$-equivariant orthogonal error bases are as follows:
\begin{itemize}
\item for orbit type (2,1,1), two isolated solutions;
\item for orbit type (2,2), two isolated solutions.
\end{itemize}
\end{proposition}

\begin{proof}
Any solution for $D_4$ must also be a solution for its $\mathbb{Z}_4$ subgroup. In Proposition \ref{Z4equivariantrotations} we saw that these form a single 2-parameter family; they can only be preserved if $f=x$ or $f=x+y$, and the points on the $z$-axis are either $\{0,\pi\}$, which yields orbit type $(2,1,1)$, or $\{-\pi/2,\pi/2\}$, which yields orbit type $(2,2)$.
\end{proof}

\begin{proposition}
There are no $D_n$-equivariant  orthogonal error bases for $n \geq 5$.
\end{proposition}

\begin{proof}
If there is no equivariant OEB for the cyclic subgroup there can be none for the full dihedral group. The result therefore follows from Proposition~\ref{cyclicqubituebs}.
\end{proof}

\subsubsection{Other subgroups of \SO(3)}

\begin{proposition}
\label{tetrahedraluebs}
The tetrahedral subgroups have two equivariant  orthogonal error bases, both of orbit type~$(4)$.
\end{proposition}

\begin{proof}
Any solution for the tetrahedral group must also be a solution for its $\mathbb{Z}_3$ subgroup. These form a 2-parameter family of tetrahedra. Since the tetrahedral group preserves only a regular tetrahedron and its dual, there will be exactly two solutions corresponding to the vertices of those tetrahedra. 
\end{proof}

\begin{proposition}
\label{octahedraluebs}
The octahedral subgroups have one equivariant orthogonal error basis, of orbit type~$(1,3)$. 
\end{proposition}

\begin{proof}
Any solution for the octahedral group must also be a solution for its $D_4$ subgroup. Only one of these equivariant for the full octahedral group, with three points at the face centres of a cube of centre-to-face length $\pi$, and the final point at the origin. 
\end{proof}

\begin{proposition}
The icosahedral subgroups have no equivariant orthogonal error bases.
\end{proposition}
\begin{proof}
There is no equivariant  OEB for the $D_5$ subgroup, so there will be none for the full icosahedral group.
\end{proof}

\subsection{Higher dimensions} 
In this section we consider the problem of constructing an equivariant UEB for representations of dimension greater than two.
\subsubsection{Constructions for permutation representations}\label{qudituebconstructions}
Recall that a representation $\rho:G \to U(n)$ is a \emph{permutation representation} if there exists an orthonormal basis of $\mathbb{C}^n$ in which $\rho(g)$, $g \in G$ are all permutation matrices.  In this special case, equivariant UEBs can be constructed from Hadamard matrices satisfying a certain equivariance condition.
\begin{proposition} \label{Hadamardconstruction}
Let $(G,\rho)$ be a permutation representation, and let $H$ be a Hadamard matrix that commutes with all permutation matrices $\rho(g)$. Then the following are elements of a $G$-equivariant unitary error basis: 
\begin{equation}
(U_H)_{ij} = \frac{1}{N}H \circ \diag(H,j)^{\dagger} \circ H^{\dagger} \circ \diag(H^T,i)
\end{equation}
Here $diag(M,i)$ is the diagonal matrix whose diagonal is the $i$th row of $M$.
\end{proposition}

\begin{proof}
It is proven in~\cite[Corollary 35]{Musto2016} that this is a UEB; showing $G$-equivariance is a simple exercise in matrix algebra. 
\end{proof}
\noindent
We will use this construction to prove Theorem~\ref{constructiontheorem}. First we need the following lemma.
\begin{lemma}\label{unitarymatrixconditions}
Let $M$ be a circulant matrix of dimension $\geq 3$ whose  first column vector $(a,b,\dots,b)$ has first entry $a$ and all other entries $b$. Let $a=|a|\alpha, b= |b|\beta $ where $\alpha,\beta \in U(1)$ and $|a|,|b| \neq 0$. Then $M$ is unitary precisely when the following conditions are satisfied:

\noindent\begin{minipage}{0.3\linewidth}
\begin{equation}
\label{inequalityfora} \frac{n-2}{n} \leq |a| \leq 1
\end{equation}
\end{minipage}%
\begin{minipage}{0.33\linewidth}
\begin{equation}
 \label{equalityforb} |b|^2 = \frac{1 - |a|^2}{n-1}
\end{equation}
\end{minipage}%
\begin{minipage}{0.36\linewidth}
\begin{equation}
 \Re(\alpha^* \beta) = \frac{2-n}{2} \frac{|b|}{|a|}
\end{equation}
\end{minipage}\par\vspace{0pt}

\end{lemma}

\begin{proof}
For unitarity it is sufficient that the rows form an orthonormal basis. It is clear from the symmetry of $M$ that it is sufficient for one row vector to be normal, and one pair of row vectors to be orthogonal. This gives us two equations in $a$ and~$b$:
\begin{align}
|b|^2 &= \frac{1 - |a|^2}{n-1}  \label{unq1} \\
\Re(a^*b) &= \frac{2-n}{2} |b|^2. \label{unq2}   
\end{align}
We will demonstrate that (\ref{inequalityfora}) is necessary and sufficient  for us to find $b$ satisfying these equations. It is obvious that (\ref{unq1}) is satisfiable if and only if $|a| \leq 1$.
Letting $a=|a|\alpha, b= |b|\beta$, equation (\ref{unq2}) then reads as follows: 
$$
\Re(\alpha^* \beta) = \frac{2-n}{2} \frac{|b|}{|a|}
$$
Since $-1 \leq \Re(\alpha^* \beta) \leq 1$, and $\alpha, \beta$ can be freely adjusted to give $\Re(\alpha^* \beta)$ any value in that range, we see that the following is necessary and sufficient for (\ref{unq2}) to be soluble:
$$ \frac{(2-n)^2}{4} \frac{|b|^2}{|a|^2} \leq  1 $$
Use of the equation ($\ref{unq1}$) and a short calculation demonstrates that this is equivalent to the lower bound in the inequality (\ref{inequalityfora}).
\end{proof}

\begin{theorem}\label{constructiontheorem}
There exists a $G$-equivariant unitary error basis for every permutation representation $(G,\rho)$ of dimension less than 5.
\end{theorem}

\begin{proof}
We use the construction in Proposition~\ref{Hadamardconstruction}.
Expressed in the $G$-permuted orthonormal basis, $\Im(\rho)$ will be some subgroup of the permutation matrices $S_n$. To use Theorem \ref{Hadamardconstruction}, we must find a Hadamard matrix in the centraliser of $\rho(G)$. In the worst case, $\Im(\rho)$ will be all permutation matrices. 

For dimension less than 5, there exists a Hadamard matrix which commutes with all permutation matrices. We ignore the degenerate case $n=1$. For $n=2$ the following family of Hadamard matrices commutes with $S_2$, where $|a|=|b|=1/\sqrt{2}$ and $Re (a^*b)=0$:
$$
\begin{pmatrix}
a & b \\
b & a
\end{pmatrix}
$$
For $n \geq 3$, the centraliser of $S_n$ is the group of circulant matrices of the type described in Lemma~\ref{unitarymatrixconditions}; the conditions for such a matrix to be unitary were given there.
Setting $|a|=|b|$ in (\ref{equalityforb}), it follows that
$|a| = 1/\sqrt{n}.$ This is compatible with (\ref{inequalityfora}) only for $n \leq 4$.
\end{proof}

\subsubsection{Showing nonexistence}
\label{sec:existence}
In this section we provide a method for proving nonexistence of an equivariant unitary error basis for some representations $(G,\rho)$.

\begin{definition}\label{geqonbdefn}
A representation $\rho: G \to U(n)$ on some $n$-dimensional vector space $V$ is \emph{monomial}~\cite{Curtis1966} if it admits an orthonormal basis of $\mathbb{C}^n$ in which all the matrices $\rho(g), g \in G$ are monomial.
\end{definition}

$G$-equivariant unitary  error bases for $(G,\rho)$ are $G$-equivariant orthonormal bases of $\End(V) \simeq \rho \otimes \rho^*$, all of whose elements are unitary maps. Therefore, if $(G,\rho)$ admits an equivariant UEB, then $\rho \otimes \rho^*$ must be monomial. It is also well known~\cite{Curtis1966} that every monomial representation is a direct sum of representations induced from one-dimensional representations of subgroups. We therefore obtain the following proposition.
\begin{proposition}
If $(G,\rho)$ admits an equivariant UEB, then $\rho \otimes \rho^*$ must split as a direct sum of representations induced from one-dimensional representations of subgroups.
\end{proposition}
This condition is straightforward to check using characters in a computer algebra program such as GAP~\cite{GAP2016}. As an example, we exhibit a 3-dimensional representation for which no equivariant UEBs exist.
\begin{example}\label{noa5gequeb}
We show that the 3\-dimensional irreducible representations of the alternating group $A_5$ admit no equivariant unitary error basis. In Table~\ref{chartablea5} are shown the characters of the induced monomial representations of the alternating group $A_5$ of dimension less than or equal to 9.
\begin{table}
\centeringforarxiv{\version}
\caption{Simple monomial representations for $A_5$.}
\label{chartablea5}
\begin{tabular}{ccccc}
\hline
\hline
 $()$ & $(1,2) (3,4)$ & $(1,2,3)$ & $(1,2,3,4,5)$ & $(1,2,3,5,4)$ \\
\hline
 1 & 1 & 1 & 1 & 1    \\
 5 & 1 & -1 & 0 & 0   \\
 5 & 1 & 2 & 0 & 0    \\
 6 & -2 & 0 & 1 & 1   \\
 6 & 2 & 0 & 1 & 1      \\
 \hline
 \hline
\end{tabular}
\end{table}
We see that $\chi_{V_i}(1,2,3,4,5) = (\pm \sqrt{5}+1)/2$; this means that $\chi_{V_i \otimes V_i^*}(1,2,3,4,5)$  has a multiple of $\sqrt{5}$ as a summand for both of $i=1,2$. However, all the monomial characters of $A_5$ of degree less than $9$ have integer values. $\chi_{V_i \otimes V_i^*}$ can therefore not be decomposed as a $\mathbb{Z}_+$-linear combination of monomial characters.
\end{example}
\finitepraacknowledgements{\version}


\biblstyle{\version}
\bibliography{FiniteGroupTeleportation}

\begin{thebibliography}{10}

\bibitem{Armstrong1997}
Margaret~A. Armstrong.
\newblock {\em Groups and Symmetry}.
\newblock Undergraduate Texts in Mathematics. Springer New York, 1997.

\bibitem{Bacsardi2013}
Laszlo Bacsardi.
\newblock On the way to quantum-based satellite communication.
\newblock {\em IEEE Communications Magazine}, 51(8):50--55, 2013.
\newblock \href {http://dx.doi.org/10.1109/MCOM.2013.6576338}
  {\path{doi:10.1109/MCOM.2013.6576338}}.

\bibitem{Bartlett2004}
Stephen~D. Bartlett, Terry Rudolph, and Robert~W. Spekkens.
\newblock Decoherence-full subsystems and the cryptographic power of a private
  shared reference frame.
\newblock {\em Physical Review A}, 70:032307, 2004.
\newblock \href {http://arxiv.org/abs/quant-ph/0403161}
  {\path{arXiv:quant-ph/0403161}}, \href
  {http://dx.doi.org/10.1103/PhysRevA.70.032307}
  {\path{doi:10.1103/PhysRevA.70.032307}}.

\bibitem{Bartlett2007}
Stephen~D. Bartlett, Terry Rudolph, and Robert~W. Spekkens.
\newblock Reference frames, superselection rules, and quantum information.
\newblock {\em Reviews in Modern Physics}, 79:555--609, 2007.
\newblock \href {http://arxiv.org/abs/quant-ph/0610030}
  {\path{arXiv:quant-ph/0610030}}, \href
  {http://dx.doi.org/10.1103/RevModPhys.79.555}
  {\path{doi:10.1103/RevModPhys.79.555}}.

\bibitem{Bennett1993}
Charles~H. Bennett, Gilles Brassard, Claude Cr\'epeau, Richard Jozsa, Asher
  Peres, and William~K. Wootters.
\newblock Teleporting an unknown quantum state via dual classical and
  {E}instein-{P}odolsky-{R}osen channels.
\newblock {\em Physical Review Letters}, 70:1895--1899, 1993.
\newblock \href {http://dx.doi.org/10.1103/PhysRevLett.70.1895}
  {\path{doi:10.1103/PhysRevLett.70.1895}}.

\bibitem{Chiribella2012}
Giulio Chiribella, Vittorio Giovannetti, Lorenzo Maccone, and Paolo Perinotti.
\newblock Teleportation transfers only speakable quantum information.
\newblock {\em Physical Review A}, 86:010304, 2012.
\newblock \href {http://arxiv.org/abs/1008.0967} {\path{arXiv:1008.0967}},
  \href {http://dx.doi.org/10.1103/PhysRevA.86.010304}
  {\path{doi:10.1103/PhysRevA.86.010304}}.

\bibitem{Curtis1966}
Charles~W. Curtis and Irving Reiner.
\newblock {\em Representation Theory of Finite Groups and Associative
  Algebras}.
\newblock AMS Chelsea Publishing Series. Interscience, 1966.

\bibitem{DAmbrosio2012}
Valerio D'Ambrosio, Eleonora Nagali, Stephen~P. Walborn, Leandro Aolita, Sergei
  Slussarenko, Lorenzo Marrucci, and Fabio Sciarrino.
\newblock Complete experimental toolbox for alignment-free quantum
  communication.
\newblock {\em Nature Communications}, 3:961, 2012.
\newblock \href {http://arxiv.org/abs/1203.6417} {\path{arXiv:1203.6417}},
  \href {http://dx.doi.org/10.1038/ncomms1951} {\path{doi:10.1038/ncomms1951}}.

\bibitem{Duligall2006}
J.~L. Duligall, M.~S. Godfrey, K.~A. Harrison, W.~J. Munro, and J.~G. Rarity.
\newblock Low cost and compact quantum key distribution.
\newblock {\em New Journal of Physics}, 8(10):249, 2006.
\newblock \href {http://arxiv.org/abs/quant-ph/0608213}
  {\path{arXiv:quant-ph/0608213}}, \href
  {http://dx.doi.org/10.1088/1367-2630/8/10/249}
  {\path{doi:10.1088/1367-2630/8/10/249}}.

\bibitem{Duligall2007}
J.~L. Duligall, M.~S. Godfrey, A.~M. Lynch, W.~J. Munro, K.~J. Harrison, and
  J.~G. Rarity.
\newblock Low cost quantum secret key growing for consumer transactions.
\newblock In {\em CLEO Europe and IQEC 2007 Conference Digest}. Optical Society
  of America, 2007.
\newblock \href {http://dx.doi.org/10.1364/IQEC.2007.JSI2_4}
  {\path{doi:10.1364/IQEC.2007.JSI2_4}}.

\bibitem{Euler1776}
Leonhard Euler.
\newblock Formulae generales pro translatione quacunque corporum rigidorum.
\newblock {\em Novi Commentarii academiae scientiarum Petropolitanae},
  20:189--207, 1776.

\bibitem{GAP2016}
The GAP~Group.
\newblock {\em {GAP -- Groups, Algorithms, and Programming, Version 4.8.6}},
  2016.
\newblock URL: \url{http://www.gap-system.org}.

\bibitem{Gisin2002}
Nicolas Gisin, Gr\'egoire Ribordy, Wolfgang Tittel, and Hugo Zbinden.
\newblock Quantum cryptography.
\newblock {\em Reviews of Modern Physics}, 74:145--195, 2002.
\newblock \href {http://arxiv.org/abs/quant-ph/0101098}
  {\path{arXiv:quant-ph/0101098}}, \href
  {http://dx.doi.org/10.1103/RevModPhys.74.145}
  {\path{doi:10.1103/RevModPhys.74.145}}.

\bibitem{Gour2008}
Gilad Gour and Robert~W Spekkens.
\newblock The resource theory of quantum reference frames: manipulations and
  monotones.
\newblock {\em New Journal of Physics}, 10(3):033023, 2008.
\newblock \href {http://arxiv.org/abs/0711.0043} {\path{arXiv:0711.0043}},
  \href {http://dx.doi.org/10.1088/1367-2630/10/3/033023}
  {\path{doi:10.1088/1367-2630/10/3/033023}}.

\bibitem{Ioannou2014}
Lawrence~M. Ioannou and Michele Mosca.
\newblock Public-key cryptography based on bounded quantum reference frames.
\newblock {\em Theoretical Computer Science}, 560(P1):33--45, 2014.
\newblock \href {http://arxiv.org/abs/0903.5156} {\path{arXiv:0903.5156}},
  \href {http://dx.doi.org/10.1016/j.tcs.2014.09.016}
  {\path{doi:10.1016/j.tcs.2014.09.016}}.

\bibitem{Islam2014}
Tanvirul {Islam}, Lo\"ick {Magnin}, Brandon {Sorg}, and Stephanie {Wehner}.
\newblock {Spatial reference frame agreement in quantum networks}.
\newblock {\em New Journal of Physics}, 16(6):063040, 2014.
\newblock \href {http://arxiv.org/abs/1306.5295} {\path{arXiv:1306.5295}},
  \href {http://dx.doi.org/10.1088/1367-2630/16/6/063040}
  {\path{doi:10.1088/1367-2630/16/6/063040}}.

\bibitem{Islam2016}
Tanvirul Islam and Stephanie Wehner.
\newblock Asynchronous reference frame agreement in a quantum network.
\newblock {\em New Journal of Physics}, 18(3):033018, 2016.
\newblock \href {http://arxiv.org/abs/1505.02565} {\path{arXiv:1505.02565}},
  \href {http://dx.doi.org/10.1088/1367-2630/18/3/033018}
  {\path{doi:10.1088/1367-2630/18/3/033018}}.

\bibitem{Kitaev2004}
Alexei Kitaev, Dominic Mayers, and John Preskill.
\newblock Superselection rules and quantum protocols.
\newblock {\em Physical Review A}, 69:052326, 2004.
\newblock \href {http://arxiv.org/abs/quant-ph/0310088}
  {\path{arXiv:quant-ph/0310088}}, \href
  {http://dx.doi.org/10.1103/PhysRevA.69.052326}
  {\path{doi:10.1103/PhysRevA.69.052326}}.

\bibitem{Laing2010}
Anthony Laing, Valerio Scarani, John~G. Rarity, and Jeremy~L. O'Brien.
\newblock Reference-frame-independent quantum key distribution.
\newblock {\em Physical Review A}, 82:012304, 2010.
\newblock \href {http://arxiv.org/abs/1003.1050} {\path{arXiv:1003.1050}},
  \href {http://dx.doi.org/10.1103/PhysRevA.82.012304}
  {\path{doi:10.1103/PhysRevA.82.012304}}.

\bibitem{Liang2014}
Wen-Ye Liang, Shuang Wang, Hong-Wei Li, Zhen-Qiang Yin, Wei Chen, Yao Yao,
  Jing-Zheng Huang, Guang-Can Guo, and Zheng-Fu Han.
\newblock Proof-of-principle experiment of reference-frame-independent quantum
  key distribution with phase coding.
\newblock {\em Scientific Reports}, 4:3617, 2014.
\newblock \href {http://arxiv.org/abs/1405.2136} {\path{arXiv:1405.2136}},
  \href {http://dx.doi.org/10.1038/srep03617} {\path{doi:10.1038/srep03617}}.

\bibitem{Marzolino2015}
Ugo Marzolino and Andreas Buchleitner.
\newblock Quantum teleportation with identical particles.
\newblock {\em Physical Review A}, 91:032316, 2015.
\newblock \href {http://arxiv.org/abs/1502.05814} {\path{arXiv:1502.05814}},
  \href {http://dx.doi.org/10.1103/PhysRevA.91.032316}
  {\path{doi:10.1103/PhysRevA.91.032316}}.

\bibitem{Marzolino2016}
Ugo Marzolino and Andreas Buchleitner.
\newblock Performances and robustness of quantum teleportation with identical
  particles.
\newblock {\em Proceedings of the Royal Society of London A: Mathematical,
  Physical and Engineering Sciences}, 472(2185), 2016.
\newblock \href {http://arxiv.org/abs/1512.02692} {\path{arXiv:1512.02692}},
  \href {http://dx.doi.org/10.1098/rspa.2015.0621}
  {\path{doi:10.1098/rspa.2015.0621}}.

\bibitem{Musto2016}
Benjamin Musto and Jamie Vicary.
\newblock Quantum {L}atin squares and unitary error bases.
\newblock {\em Quantum Information {\&} Computation}, 16(15{\&}16):1318--1332,
  2016.
\newblock \href {http://arxiv.org/abs/1504.02715} {\path{arXiv:1504.02715}}.

\bibitem{Nielsen2011}
Michael~A. Nielsen and Isaac~L. Chuang.
\newblock {\em Quantum Computation and Quantum Information: 10th Anniversary
  Edition}.
\newblock Cambridge University Press, New York, NY, USA, 10th edition, 2011.

\bibitem{Peres2002}
Asher {Peres} and Petra~F. {Scudo}.
\newblock Unspeakable quantum information.
\newblock In A~Khrennikov, editor, {\em Quantum Theory: Reconsideration of
  Foundations}. Va\"xjo Univesity Press, 2002.
\newblock \href {http://arxiv.org/abs/quant-ph/0201017}
  {\path{arXiv:quant-ph/0201017}}.

\bibitem{Ren2017}
Ji-Gang Ren, Ping Xu, Hai-Lin Yong, Liang Zhang, Sheng-Kai Liao, Juan Yin,
  Wei-Yue Liu, Wen-Qi Cai, Meng Yang, Li~Li, et~al.
\newblock Ground-to-satellite quantum teleportation.
\newblock {\em Nature}, 549(7670):70--73, 2017.
\newblock \href {http://arxiv.org/abs/1707.00934} {\path{arXiv:1707.00934}},
  \href {http://dx.doi.org/10.1038/nature23675}
  {\path{doi:10.1038/nature23675}}.

\bibitem{Skotiniotis2013}
Michael Skotiniotis, Wolfgang D\"{u}r, and Barbara Kraus.
\newblock Efficient quantum communication under collective noise.
\newblock {\em Quantum Information and Computation}, 13(3\&4):0290--0323, 2013.
\newblock \href {http://arxiv.org/abs/1204.0891} {\path{arXiv:1204.0891}}.

\bibitem{Skotiniotis2012}
Michael Skotiniotis and Gilad Gour.
\newblock Alignment of reference frames and an operational interpretation for
  the {G}-asymmetry.
\newblock {\em New Journal of Physics}, 14(7):073022, 2012.
\newblock \href {http://arxiv.org/abs/1202.3163} {\path{arXiv:1202.3163}},
  \href {http://dx.doi.org/10.1088/1367-2630/14/7/073022}
  {\path{doi:10.1088/1367-2630/14/7/073022}}.

\bibitem{Souza2008}
C.~E.~R. Souza, C.~V.~S. Borges, A.~Z. Khoury, J.~A.~O. Huguenin, L.~Aolita,
  and S.~P. Walborn.
\newblock Quantum key distribution without a shared reference frame.
\newblock {\em Physical Review A}, 77:032345, 2008.
\newblock \href {http://dx.doi.org/10.1103/PhysRevA.77.032345}
  {\path{doi:10.1103/PhysRevA.77.032345}}.

\bibitem{Vallone2014}
Giuseppe Vallone, Vincenzo D'Ambrosio, Anna Sponselli, Sergei Slussarenko,
  Lorenzo Marrucci, Fabio Sciarrino, and Paolo Villoresi.
\newblock Free-space quantum key distribution by rotation-invariant twisted
  photons.
\newblock {\em Physical Review Letters}, 113:060503, 2014.
\newblock \href {http://arxiv.org/abs/1402.2932} {\path{arXiv:1402.2932}},
  \href {http://dx.doi.org/10.1103/PhysRevLett.113.060503}
  {\path{doi:10.1103/PhysRevLett.113.060503}}.

\bibitem{Enk2001}
Steven van Enk.
\newblock The physical meaning of phase and its importance for quantum
  teleportation.
\newblock {\em Journal of Modern Optics}, 48(13):2049--2054, 2001.
\newblock \href {http://arxiv.org/abs/quant-ph/0102004}
  {\path{arXiv:quant-ph/0102004}}, \href
  {http://dx.doi.org/10.1080/09500340108240906}
  {\path{doi:10.1080/09500340108240906}}.

\bibitem{Verdon2017}
Dominic Verdon and Jamie Vicary.
\newblock Tight reference frame--independent quantum teleportation.
\newblock {\em Electronic Proceedings in Theoretical Computer Science},
  236:202--214, 2017.
\newblock \href {http://arxiv.org/abs/1603.08866v1}
  {\path{arXiv:1603.08866v1}}, \href {http://dx.doi.org/10.4204/EPTCS.236.13}
  {\path{doi:10.4204/EPTCS.236.13}}.

\bibitem{VerdonInfinite}
Dominic Verdon and Jamie Vicary.
\newblock Quantum teleportation with infinite reference frame uncertainty.
\newblock In preparation, 2018.
\newblock \href {http://arxiv.org/abs/1802.09040} {\path{arXiv:1802.09040}}.

\bibitem{Verstraete2003}
Frank Verstraete and Juan~I. Cirac.
\newblock Quantum nonlocality in the presence of superselection rules and data
  hiding protocols.
\newblock {\em Physical Review Letters}, 91:010404, 2003.
\newblock \href {http://arxiv.org/abs/quant-ph/0302039}
  {\path{arXiv:quant-ph/0302039}}, \href
  {http://dx.doi.org/10.1103/PhysRevLett.91.010404}
  {\path{doi:10.1103/PhysRevLett.91.010404}}.

\bibitem{Wabnig2013}
J~Wabnig, D~Bitauld, H~W Li, A~Laing, J~L O'Brien, and A~O Niskanen.
\newblock Demonstration of free-space reference frame independent quantum key
  distribution.
\newblock {\em New Journal of Physics}, 15(7):073001, 2013.
\newblock \href {http://arxiv.org/abs/1305.0158} {\path{arXiv:1305.0158}},
  \href {http://dx.doi.org/10.1088/1367-2630/15/7/073001}
  {\path{doi:10.1088/1367-2630/15/7/073001}}.

\bibitem{Wang2015}
Can {Wang}, Shi-Hai {Sun}, Xiang-Chun {Ma}, Guang-Zhao {Tang}, and Lin-Mei
  {Liang}.
\newblock {Reference-frame-independent quantum key distribution with source
  flaws}.
\newblock {\em Physical Review A}, 92(4):042319, 2015.
\newblock \href {http://dx.doi.org/10.1103/PhysRevA.92.042319}
  {\path{doi:10.1103/PhysRevA.92.042319}}.

\bibitem{Werner2001}
Reinhard~F. Werner.
\newblock All teleportation and dense coding schemes.
\newblock {\em Journal of Physics A: Mathematical and General}, 34(35):7081,
  2001.
\newblock \href {http://arxiv.org/abs/quant-ph/0003070}
  {\path{arXiv:quant-ph/0003070}}, \href
  {http://dx.doi.org/10.1088/0305-4470/34/35/332}
  {\path{doi:10.1088/0305-4470/34/35/332}}.

\bibitem{Yin2017}
Juan Yin, Yuan Cao, Yu-Huai Li, Sheng-Kai Liao, Liang Zhang, Ji-Gang Ren,
  Wen-Qi Cai, Wei-Yue Liu, Bo~Li, et~al.
\newblock Satellite-based entanglement distribution over 1200 kilometers.
\newblock {\em Science}, 356(6343):1140--1144, 2017.
\newblock \href {http://arxiv.org/abs/1707.01339} {\path{arXiv:1707.01339}},
  \href {http://dx.doi.org/10.1126/science.aan3211}
  {\path{doi:10.1126/science.aan3211}}.

\bibitem{Zhang2014}
P.~Zhang, K.~Aungskunsiri, E.~Mart\'{\i}n-L\'opez, J.~Wabnig, M.~Lobino, R.~W.
  Nock, J.~Munns, et~al.
\newblock Reference-frame-independent quantum-key-distribution server with a
  telecom tether for an on-chip client.
\newblock {\em Physical Review Letters}, 112:130501, 2014.
\newblock \href {http://arxiv.org/abs/1308.3436} {\path{arXiv:1308.3436}},
  \href {http://dx.doi.org/10.1103/PhysRevLett.112.130501}
  {\path{doi:10.1103/PhysRevLett.112.130501}}.

\end{thebibliography}


\appendix

\section{Existence of $G$-invariant maximally entangled states}
\label{sec:invariantstates}

\noindent Here we prove the result stated in Remark~\ref{rem:entangledstatealignment}.

\begin{definition}
A state $\omega$ of a $G$-representation is \textit{invariant up to a phase} if $g \cdot \omega = \theta(g) \omega$ for some homomorphism $\theta: G \to U(1)$.
\end{definition}

\begin{lemma}
Let $A,B$ be $G$-representations of identical dimension. A maximally entangled pure state $\omega \in A \otimes B$ invariant up to a phase exists iff $A \simeq \theta \otimes B^*$ for some $\theta:G \to U(1)$.
\end{lemma}

\begin{proof}
Suppose the representation $A$ is the dual of $B$ up to a character $\theta$. Then let $\omega$ be the unit $\eta: \mathbbm{1} \to \theta^* \otimes A \otimes B$ witnessing the duality $\theta^* \otimes A \simeq B^*$. In the other direction, suppose there exists a state stabilised up to a phase. Any maximally entangled state is of the form $$ \sum_i \ket{i} \otimes X \ket{i} $$ for some orthonormal basis $\{ \ket{i} \}$ and unitary $X$. Working in that basis we have the following, for all $g \in G$, and where $\rho_A(g)^T$ is the transpose in the basis $\{\ket{i}\}$:
\begin{align*}
g \cdot \sum_i \ket{i} \otimes V \ket{i} &= \sum_i \rho_A(g) \ket{i} \otimes \rho_B(g) V \ket{i}  \\
&= \sum_i \ket{i} \otimes \rho_B(g) V \rho_A(g)^T \ket{i}
\end{align*}
It follows that $\rho_B(g) V \rho_A(g)^T = \theta(g) V$, and therefore that $\rho_B(g) = \theta(g) V \rho_A(g)^* V^\dagger$ for all $g$, where $\rho_A(g)^*$ is the complex conjugate matrix. The result follows by definition of the dual representation.
\end{proof}

\end{document}